\documentclass[12pt,a4paper]{article}

\usepackage{amsmath,amsthm,amssymb,amscd,a4wide}

\usepackage{dsfont}
\usepackage{mathrsfs}
\usepackage{setspace}
\usepackage{url}
\usepackage{color}

\numberwithin{equation}{section}

\DeclareMathOperator{\op}{op}
\DeclareMathOperator{\supp}{supp}
\DeclareMathOperator{\sgn}{sgn}
\DeclareMathOperator{\Hess}{Hess}
\DeclareMathOperator{\Mat}{Mat}

\newcommand{\ud}{\mathrm{d}}
\newcommand{\ui}{\mathrm{i}}
\newcommand{\ue}{\mathrm{e}}
\newcommand{\Or}{O}

\DeclareMathOperator{\vol}{vol}
\DeclareMathOperator{\C}{C}
\DeclareMathOperator{\diam}{diam}
\DeclareMathOperator{\tr}{tr}
\DeclareMathOperator{\T}{T}
\DeclareMathOperator{\HS}{HS}

\newcommand{\gz}{\mathbb Z}
\newcommand{\fz}{\mathbb F}
\newcommand{\tz}{\mathbb T}

\newcommand{\rz}{\mathbb R}
\newcommand{\nz}{\mathbb N}
\newcommand{\kz}{\mathbb C}

\newcommand{\be}{\begin{equation}}
\newcommand{\ee}{\end{equation}}
\newcommand{\lk}{\left(}

\newcommand{\rk}{\right)}
\newcommand{\eins}{\mathds{1}}

\newtheorem{theorem}{Theorem}
\newtheorem{lemma}{Lemma}
\newtheorem{prop}{Proposition}
\newtheorem{cor}{Corollary}
\newtheorem{defn}{Definition}

\usepackage{color}

\begin{document}

\thispagestyle{empty}

\begin{center}

{\LARGE\bf A Gutzwiller trace formula for}\\[3ex] 
{\LARGE\bf large hermitian matrices}\\

\vspace*{3cm}

{\large Jens Bolte\footnote{Department of Mathematics, Royal Holloway,
University of London, Egham, TW20 0EX, United Kingdom, 
{\tt jens.bolte@rhul.ac.uk}}, 
{\large Sebastian Egger}\footnote{Department of Mathematics,
Technion--Israel Institute of Technology, 629 Amado Building, Haifa 32000, 
Israel, {\tt egger@tx.technion.ac.il }}, and 
Stefan Keppeler\footnote{Fachbereich Mathematik, Universit\"at T\"ubingen, 
Auf der Morgenstelle 10, 72076 T\"ubingen, Germany, 
{\tt stefan.keppeler@uni-tuebingen.de}}}
\end{center}
\vfill
\begin{abstract}
We develop a semiclassical approximation for the dynamics of quantum
systems in finite-dimensional Hilbert spaces whose classical
counterparts are defined on a toroidal phase space. In contrast to
previous models of quantum maps, the time evolution is in continuous
time and, hence, is generated by a Schr\"odinger equation. In the
framework of Weyl quantisation, we construct discrete, semiclassical
Fourier integral operators approximating the unitary time evolution
and use these to prove a Gutzwiller trace formula. We briefly discuss
a semiclassical quantisation condition for eigenvalues as well as some
simple examples.
\end{abstract}

\newpage

\section{Introduction}
\label{intro}
Semiclassical analysis aims at approximating quantum dynamics in terms
of suitable, corresponding classical dynamics in the semiclassical
limit. Since the unitary quantum dynamics are generated by a
self-adjoint Hamiltonian operator, a closely related problem for
operators with a discrete spectrum is the semiclassical approximation
of eigenvalues.  Bohr-Sommerfeld quantisation conditions in one
dimension were the first examples of such approximations. They were
followed by the EBK conditions for integrable systems. A general
approximation of quantum spectral functions in terms of classical
quantities, however, was first given in terms of the Gutzwiller trace
formula \cite{Gutzwiller:1971}.  The first rigorous proofs of a trace
formula in a similar spirit were given by Colin~de~Verdi{\`e}re
\cite{CdV:1973a,CdV:1973b}, and Duistermaat and Guillemin
\cite{Duistermaat:1975}, for the case of a Laplacian on a closed
manifold. The semiclassical case originally considered by Gutzwiller
was proven by Meinrenken \cite{Meinrenken:1992}.

Semiclassical methods have been successfully applied in a broad range
of circumstances, from realistic models describing experiments to
purely mathematical models. In many of those cases a classical
configuration space is given. The classical phase space then is the
cotangent bundle over the configuration space. The corresponding
quantum systems are defined in an infinite dimensional Hilbert space,
typically an $L^2$-space over the configuration space, and
semiclassical methods are available that relate the quantum and the
classical descriptions of the same physical system, see, e.g.,
\cite{Zworski:2012}. In some situations (as, e.g., Toeplitz
quantisation and quantum maps
\cite{Zelditch:1997,Esposti:2003,Schlichenmaier:2010}) the classical
phase space, however, is a compact symplectic manifold, such as a
sphere or a torus. There then exists no classical configuration
space. The associated quantum Hilbert space is finite dimensional and
the dimension of the Hilbert space tends to infinity in the
semiclassical limit. A class of models of this type that have been
studied intensively are quantised torus maps \cite{Esposti:2003}, for
which the dynamics take place in discrete time. Quantum maps are often
studied as mathematical toy models, e.g.\ in the context of quantum
chaos. However, phase spaces that are not cotangent bundles also have
many applications in physics. For instance, toroidal phase spaces are
relevant in solid state physics, see, e.g., \cite{Harper:1955}, and
more general symplectic manifolds play a role in molecular physics,
see, e.g., \cite{Sadovskii:2006} and references therein.

Our goal in this article is to quantise Hamiltonian flows on tori and
to approximate the resulting continuous-time quantum dynamics
semiclassically, within the framework of Weyl quantisation.  We then
use the semiclassical approximation of the dynamics in order to prove
a trace formula. As the quantum Hamiltonian acts in a finite
dimensional Hilbert space, this is a Gutzwiller trace formula for
hermitian matrices, with matrix size growing in the semiclassical
limit.  Since tori are special cases of compact K\"ahler manifolds, it
may seem natural to apply results that were obtained in the context of
Toeplitz quantisation \cite{Boutet:1981} on compact K\"ahler
manifolds.  In this framework a Gutzwiller-type trace formula was
proven by Borthwick, Paul and Uribe \cite{Borthwick:1998}, building on
\cite{Boutet:1981}. In a more recent approach, Paoletti uses the
Szeg\"o kernel of \cite{Boutet:1976} in order to also prove local
asymptotics \cite{Paoletti:2011,Paoletti:2016}. Nevertheless, although
the approach of the present article is not suitable for general
compact K\"ahler manifolds we believe that proving a trace formula for
torus flows entirely within the framework of Weyl quantisation is
useful for the following two reasons.

First, Weyl quantisation of observables on a torus is well known in
the context of quantum maps. Furthermore, it allows us to outline the
proof of the trace formula in close analogy to the classic proof in
the case where the phase space is a cotangent bundle. The essential
modification is a discretised ansatz for the semiclassical time
evolution operator, see Eq.~\eqref{oscint} below. Moreover, our main
tool is a result that is essentially a stationary phase theorem for
sums, see Appendix~\ref{20}.

Second, Weyl quantisation is particularly adapted to the case of
Schr\"odinger operators in finite dimension including, e.g.,
discretisations of differential operators as in the Harper model
\cite{Harper:1955} or in the case of Dirac operators in lattice field
theory, see, e.g., \cite{Rothe:2012}. Toeplitz quantisation on tori
\cite{Zelditch:1997}, however, is equivalent to anti-Wick quantisation
as can be deduced from \cite{Bouzouina:1996,Bodmann:2001}. In
representing Hamiltonian operators, anti-Wick quantisation, in turn,
is not identical to Weyl quantisation but is the same to leading order
\cite[Lemma 3.9]{Bouzouina:1996}.  In particular, whereas Weyl symbols
of Schr\"odinger operators (and similar types of difference operators)
are $\hbar$-independent and, hence, identical to their principal
symbols, anti-Wick symbols for the same operators are
$\hbar$-dependent with the Weyl symbol providing only the principal
part. We illustrate this statement in Appendix~\ref{appAW}; see also
\cite{Zelditch:2005,Faure:2015} for a comparison of different
quantisation schemes for symplectic maps.

The article is organised as follows. In
Section~\ref{sec:setting_result} we introduce the general setting and
present our main result, the trace formula in Theorem~\ref{t12}. In
Section~\ref{q27} we recall the Weyl quantisation of functions on a
torus and develop some useful expressions for the action of a Weyl
operator on a vector. In Section~\ref{f1} we review relevant aspects
of classical Hamiltonian dynamics on a torus phase space, with a
particular emphasis on extending Hamilton-Jacobi dynamics beyond
caustics and the construction of Maslov bundles. The main technical
work is presented in Section~\ref{q166}, where the semiclassical
approximation of the quantum time evolution is developed. The proof of
our main Theorem~\ref{t12} is given in Section~\ref{t1a}. An
application of the trace formula to derive semiclassical quantisation
conditions can be found in Section~\ref{sec5}. Finally, in
Section~\ref{t1} we discuss some simple examples.

\section{Setting and main result}
\label{sec:setting_result}

Let $\tz=\rz^2/(\ell_x\gz\oplus \ell_\xi\gz)$ be a two dimensional
torus, where $\ell_x ,\ell_\xi>0$ are two length parameters. The
universal covering space of $\tz$ is $\rz^2$, and as a fundamental
domain for the action of $\ell_x\gz\oplus \ell_\xi\gz$ on $\rz^2$ we
choose $\fz = \left[0,\ell_x\right) \times
\left[0,\ell_\xi\right)$.
Weyl quantisation associates an $N$-dimensional Hilbert space $\kz^N$,
with inner product
$
\langle\psi,\phi\rangle_N:=\frac{1}{N}\langle\phi,\psi\rangle_{\kz^N},
$
to the classical system, where
\begin{equation}
\label{q29a}
N=\frac{\ell_x\ell_\xi}{2\pi\hbar}
\end{equation}
is a semiclassical parameter determined by the value of Planck's
constant $\hbar$. In fact, one can view either $\hbar$ or $N^{-1}$ as
a semiclassical parameter. In the following we shall mainly use
$\hbar$ for semiclassical asymptotic expansions, but one has to
keep in mind that the relation \eqref{q29a} allows one to rewrite
expressions so as to hide $\hbar$. Also notice that since $N$ is an
integer $\hbar$ only takes values in a discrete set.

Vectors $\psi=(\psi_n)\in\kz^N$ can be seen as `wave functions' supported 
at the points 
\begin{equation}
\label{q28}
x_n=\frac{n\ell_x}{N},\quad n=0,\dots,N-1,
\end{equation}
in the interval $[0,\ell_x)$.

The Weyl quantisation of a classical observable $f\in C^\infty(\tz)$ 
is a linear operator $\op_N(f):\kz^N\to\kz^N$ defined as
\begin{equation}
\label{q34}
\op_N(f) : =\sum_{m,n\in\gz}f_{mn}\,T^{mn}.
\end{equation}
Here
\begin{equation}
\label{q30}
f_{mn} = \frac{1}{\ell_x\ell_\xi}\int_{\fz}f(x,\xi)\,
\ue^{-2\pi\ui\left(\frac{m\xi}{\ell_\xi}-\frac{nx}{\ell_x}\right)}\ud x\,\ud\xi 
\end{equation} 
are Fourier coefficients of the observable $f$, viewed as a smooth function
on the fundamental domain $\fz$, and the 
\begin{equation}
\label{Weylnm}
T^{mn}=\ue^{\ui\pi\frac{nm}{N}}T^{m0}T^{0n}
\end{equation}
are unitary operators in $\left(\kz^N,\langle\cdot,\cdot\rangle_N\right)$ 
defined through
\begin{equation}
\label{q32}
\left( T^{m0}\psi\right)_l:=\psi_{\left(l+m\right)\bmod N} \quad\text{and}\quad
\left( T^{0n}\psi\right)_l:=\ue^{-\frac{2\pi\ui ln}{N}}\psi_l. 
\end{equation}
The latter represent translations in the $x$- and $\xi$-directions,
respectively. This quantisation was introduced in \cite{Hannay:1980},
for details see, e.g., \cite{Bouzouina:1996,Esposti:2003}.

A typical operator that one would want to represent in Weyl quantisation is a 
Schr\"odinger operator $-\hbar^2\Delta +V$, where the Laplacian is a difference
operator,
\begin{equation}
\label{discLapintro}
\begin{split}
(-\Delta\psi)_l &:= -\frac{N^2}{\ell_x^2}\bigl(\psi_{(l+1)\bmod N}+\psi_{(l-1)\bmod N}-2\psi_l\bigr) \\
                       &= -\frac{N^2}{\ell_x^2}\bigl((T^{1,0}+T^{-1,0}-2\eins)\psi\bigr)_l .
\end{split}
\end{equation}
From the last expression one immediately identifies a symbol,
\begin{equation}
\label{schroedsymb}
H(x,\xi)=\frac{\ell_\xi^2}{2\pi^2}\left(1-\cos\left(\frac{2\pi \xi}{\ell_\xi}\right)\right)+V(x),
\end{equation}
such that $\op_N(H)=-\hbar^2\Delta +V$. 

Given a classical Hamiltonian $H\in C^\infty(\tz)$ the Schr\"odinger
equation is
\begin{equation}
\label{Schroeq}
\ui\hbar\frac{\ud\psi}{\ud t}(t) = \op_N(H)\psi(t),\qquad
\psi(0)=\psi_0\in\kz^N.
\end{equation}
It generates a unitary one-parameter group $U(t)$, $t\in\rz$, via
$\psi(t)=U(t)\psi_0$. (See also \cite{Ligabo:2016}, where the corresponding
Heisenberg equation is considered.) 

We remark that fixing the dimension $N$, every hermitian $N\times N$-matrix
has a representation \eqref{q34}. This is not unique in the sense that one could 
provide alternative phase space representations of the same matrix. However,
with varying $N$, the above construction yields a family of hermitian matrices
associated with a fixed phase space and a fixed function $f$. 

The first aim of this paper is to construct a semiclassical asymptotic 
expansion
\begin{equation}
\label{scevol}
U(t)\sim\sum_{k\geq 0} \left(\frac{\hbar}{\ui}\right)^k U_k(t),
\quad\hbar\to 0.
\end{equation}
In a second step we use this expansion in order to evaluate the
right-hand side of
\begin{equation}
\label{TFidea}
\sum_n\rho\left(\frac{E_n-E}{\hbar}\right) =
\frac{1}{2\pi}\int_\rz\hat\rho(t)\,\tr U(t)\,\ue^{\frac{\ui}{\hbar}Et}\ \ud t,
\end{equation}
where $\rho\in C^\infty(\rz)$ with compactly supported Fourier
transform $\hat\rho$. The sum on the left-hand side extends over all
eigenvalues $E_n$ of $\op_N(H)$, i.e., it is a spectral function
of the Hamiltonian.

Evaluating the right-hand side of \eqref{TFidea} using the
semiclassical expansion \eqref{scevol} leads to an expression which we
relate to the Hamiltonian flow on $\tz$ generated by the classical
Hamiltonian $H$. In order to achieve this we have to identify $U(t)$
as a suitable semiclassical Fourier integral operator (scFIO). Locally, 
the latter is given in terms of an oscillatory integral, whose phase
function generates the Hamiltonian flow. In the present context of
compact phase spaces this concept has to be amended appropriately,
leading to the ansatz
\begin{equation}
\label{oscint}
U_k(t)_{nm} = \frac{1}{\ell_\xi}\int_\rz a_k(t,x_n,x_m,\xi)\,
\ue^{\frac{\ui}{\hbar}\phi(t,x_n,x_m,\xi)}\ \ud\xi.
\end{equation}
Since the quantum mechanical Hilbert space has finite dimension,
the amplitude and the phase function of this `position representation'
of the time evolution operator are only evaluated at the points
\begin{equation}
\label{xdiscrete}
x_n=\frac{n\ell_x}{N},\quad n=0,\dots,N-1. 
\end{equation}
The amplitude and the phase function, however, are defined on the
universal covering space $\rz^2$, and not only on $\fz$. They will
have to be chosen such that $U(t)$ (approximately) satisfies the 
Schr\"odinger equation that follows from \eqref{Schroeq}. If the phase
function is chosen to be of the form
\begin{equation}
\phi(t,x,y,\xi) = S(t,x,\xi)-y\xi,
\end{equation}
it will turn out that $S$ is required to be a solution of the 
Hamilton-Jacobi equation,
\begin{equation}
H(x,\partial_x S) + \partial_t S =0,
\end{equation}
with initial condition $S(0,x,\xi)=x\xi$. Hence, $S$ generates the 
canonical transformation 
\begin{equation}
\left(x,\partial_xS(t,x,\xi)\right)\mapsto
\left(\partial_\xi S(t,x,\xi),\xi\right) 
\end{equation}
representing the Hamiltonian flow backwards in time (all lifted to
the covering space $\rz^2$). However, this is only true for sufficiently 
small times $t$,
i.e., as long as no caustics occur. Beyond caustics one has to piece
local, singularity free representations of the form \eqref{oscint}
together. This requires a suitable Maslov bundle, eventually
introducing Maslov phases into the resulting trace formula. We shall 
have to devote a sizeable portion of this work to solving this problem.

Assuming that $E$ is a regular value for the classical
Hamiltonian $H$, the energy surface
\begin{equation}
\label{t10}
H^{-1}(E):=\left\{(x,\xi)\in\tz; \ H(x,\xi)=E\right\}
\end{equation}
is a one-dimensional, not necessarily connected, submanifold of the
two-dimensional torus. We denote by $\mathfrak{P}_E$ the set of
periodic orbits of the Hamiltonian flow at energy $E$. The
connected components of the energy surface $H^{-1}(E)$ are the
primitive periodic orbits. Hence, the volume $\vol(H^{-1}(E))$
of the energy surface is the sum of the periods $t_{p^\#}$ of all
primitive periodic orbits $p^\#\in\mathfrak{P}_E$.  Let $W_p$ be
the action of the orbit $p$ and denote its Maslov phase by
$\sigma_p$. Our main result then is the following.
\begin{theorem}
\label{t12}
Assume that $E$ is a regular value for the classical Hamiltonian 
$H\in C^\infty(\tz)$, and let $\rho\in C^\infty(\rz)$ with Fourier transform 
$\hat{\rho}\in C_0^{\infty}(\rz)$ be a test function. Then, for every 
$p\in\mathfrak{P}_E$ there exists a function $a_p(\hbar)$ with a complete 
asymptotic expansion in powers of $\hbar$, such that
\begin{equation}
\label{t18}
\sum_n\rho\left(\frac{E_n -E}{\hbar}\right) 
= a_0(\hbar)\frac{\hat\rho(0)}{2\pi}+\sum_{p\in\mathfrak{P}_E}
\frac{\hat{\rho}\left(t_p\right)a_p(\hbar)}{2\pi}
\ue^{\frac{\ui}{\hbar}W_p-\frac{\ui\pi}{2}\sigma_p} + O(\hbar^\infty).
\end{equation}
The leading asymptotic behaviour of the amplitude functions is
\begin{equation}
\begin{split}
\label{t13}
a_0(\hbar) &= \vol\left(H^{-1}(E)\right)+O(\hbar) \\
a_p(\hbar) &= t_{p^{\#}}+O(\hbar),
\end{split}
\end{equation}
where $p^\#$ is the primitive periodic orbit associated with 
$p\in\mathfrak{P}_E$.
\end{theorem}
In the following we provide a self contained proof of this theorem
using the explicit Weyl quantisation outlined above.


%
\section{Weyl quantisation on the torus}
\label{q27}
Weyl quantisation of classical systems with toroidal phase space is
based on irreducible unitary representations of the discrete
Heisenberg group $H_3(\gz)$, see \cite{Esposti:2003}. They are labelled 
by a positive integer $N$, the dimension of the representation. 
Similarly, Weyl quantisation of classical systems with phase space 
$\rz^2$, the universal cover of $\tz$, is based on the continuous 
Heisenberg group $H_3(\rz)$, whose irreducible unitary representations 
are labelled by the positive real parameter $\hbar$, physically Planck's 
constant divided by $2\pi$, see \cite{Folland:1989}. The representation 
of $H_3(\gz)$ generated by the
unitary operators \eqref{q32} induces a representation of $H_3(\rz)$
with $\hbar$ determined by \eqref{q29a}. Consequently, the
semiclassical limit $\hbar\to0$ corresponds to $N\to\infty$, the limit
of large matrices.

Using the definitions \eqref{Weylnm} and \eqref{q32} in \eqref{q34} we
can rewrite the action of a Weyl-quantised observable $f\in
C^\infty(\tz)$ on a vector $\psi=(\psi_k)\in\kz^N$ as
\begin{equation}
\label{q36}
\begin{split}
\left(\op_N(f)\psi\right)_k
 &=\sum_{m\in\gz}\psi_{m\bmod{N}}\ \frac{1}{\ell_\xi}\int_0^{\ell_\xi} 
   f\left(\ell_x\frac{k+m}{2N},\xi\right)\ue^{2\pi\ui(k-m)\xi/\ell_\xi}\ \ud\xi\\
 &=\sum_{m\in\gz}\psi_{m\bmod{N}}\ \frac{1}{\ell_\xi}\int_0^{\ell_\xi} 
   f\left(\frac{x_k+x_m}{2},\xi\right)\ue^{\frac{\ui}{\hbar}(x_k-x_m)\xi}\ \ud\xi \, ,
\end{split}
\end{equation} 
where $m\bmod{N}$ denotes the smallest non-negative integer $m'$ such
that $N$ divides $m-m'$. The last line of \eqref{q36} follows from
\eqref{q28} and \eqref{q29a} and the result is similar to the
corresponding expression for the application of a Weyl operator in
$L^2(\rz)$.

The infinite sum in \eqref{q36} can be turned into a finite sum plus remainder
using the following result.
\begin{lemma}
\label{q39}
Let $f\in C^\infty(\tz)$ and $\psi=(\psi_k)\in\kz^N$. Then, for all $M\in\nz$, 
we have
\begin{equation}
\label{q40}
\begin{split}
\left(\op_N(f)\psi\right)_k 
  &= \sum_{|k-m|\leq N^{\frac{1}{M}}}\psi_{m\bmod{N}}\ \frac{1}{\ell_\xi}\int_0^{\ell_\xi} 
     f\left(\ell_x\frac{k+m}{2N},\xi\right)\ue^{2\pi\ui(k-m)\xi/\ell_\xi}\ \ud\xi\\
  &\qquad+O_M\left( \hbar^{\infty}\right).
\end{split}
\end{equation}
\end{lemma}
In \eqref{q40} the notation $g(\hbar)=O_M\left(\hbar^{\infty}\right)$
means that there exists a constant $C_M$, depending on $M$, such that
$|g(\hbar)|\leq C_M \hbar^\alpha$ for all $\alpha>0$.
\begin{proof}
The claim follows from integrating by parts, see e.g.\ 
\cite[Theorem 3.2.9]{Grafakos:2008}, and only requires bounds on the derivatives
of $f$ with respect to $\xi$.
\end{proof}
We remark that the restriction $|k-m|\leq N^{1/M}$ in the summation is equivalent
to the condition 
\begin{equation}
\label{q41}
|x_k-x_m|<N^{\frac{1}{M}-1}\ell_x.
\end{equation}
The result can be rephrased as follows.
\begin{lemma}
\label{q42}
Let $f\in C^\infty(\tz)$ and $\psi=(\psi_k)\in\kz^N$. Then, for all $L,M\in\nz$,
\begin{equation}
\begin{split}
\left(\op_N(f)\psi\right)_k 
  &= \sum_{|k-m|\leq N^{\frac{1}{M}}}\psi_{m\bmod{N}}\sum_{l=0}^{L-1}\frac{1}{l!}
     \left(\frac{\hbar}{2\ui}\right)^l\frac{1}{\ell_\xi}\int_0^{\ell_\xi} 
     \partial_\xi^l\partial_x^l f(x_k,\xi)\ue^{2\pi\ui(k-m)\xi/\ell_\xi}\ \ud\xi\\
  &\qquad+O\left(N^{\frac{L}{M}-L}\right).
\end{split}
\end{equation}
\end{lemma}
\begin{proof}
We first note that a Taylor expansion together with \eqref{q41} gives
\begin{equation}
\label{Taylorexp}
f\left(\frac{x_k+x_m}{2},\xi\right) = \sum_{l=0}^{L-1}\frac{1}{l!}
\left(\frac{x_m-x_k}{2}\right)^l\partial_x^l f(x_k,\xi) + 
O\left(N^{\frac{L}{M}-L}\right).
\end{equation}
Using this identity on the right-hand side of \eqref{q40}, as well as \eqref{q29a} 
and the identity
\begin{equation}
\label{q46}
(m-k)^l\ue^{2\pi\ui(k-m)\xi/{\ell_\xi}} = \left(\frac{\ui \ell_\xi}{2\pi}\right)^l
\partial_\xi^l\ue^{2\pi\ui(k-m)\xi/{\ell_\xi}},
\end{equation}
followed by an $l$-fold integration by parts, yields the desired result.
\end{proof} 
%


%
\section{Classical dynamics on the torus}
\label{f1}
Before we can proceed to construct semiclassical approximations of the
quantum time evolution we have to provide some classical input. 
Eventually we shall construct a suitable scFIO, and for that purpose we 
need equivalents to the classical data that are used in standard FIOs 
\cite{Duistermaat:1973,Hoermander:1994} or their semiclassical counterparts 
\cite{Meinrenken:1992}.

The general setting here is that of a Hamiltonian flow $\Phi_t$ on the
phase space $\tz$ generated by a Hamiltonian vector field $X_H$
associated with a Hamiltonian function $H\in C^\infty(\tz)$. This flow
is global since $\tz$ is compact.  Due to the covering of $\tz$ by
$\rz^2$, $\pi:\rz^2\to\tz$, most expressions given below hold on
$\rz^2$ as well as on $\tz$ (at least locally). At some instances,
however, we have to pay attention to differences, making use of that
fact that the covering map provides the projection and the pull back
for switching between covering space and base space.

Working on the covering space of the torus, the canonical one- and
two-forms are given by $\theta:=\xi\ud x$ and
$\omega:=\ud\theta=\ud\xi\wedge\ud x$. Here $\theta$ is not periodic
in $\xi$, but $\omega$ can be regarded as a two-form on $\tz$ as
well.

For the purpose of semiclassical constructions it is convenient to
work on the extended phase space
$T^\ast\rz\times\tz\cong\rz^2\times\tz$.  Following an established
convention, we denote points in that space as
$(t,\tau,x,\xi)\in\rz^2\times\tz$, where $t$ is the time variable and
$-\tau$ has the meaning of an energy. In these variables the canonical
two-form is
\begin{equation}
\label{q4}
\Omega:=\ud t\wedge\ud\tau\oplus\omega,
\end{equation}
both on $\rz^2\times\rz^2$ and on $\rz^2\times\tz$. We then consider
the extended classical Hamiltonian
\begin{equation}
\label{q3}
H_{\rm ext}(t,\tau,x,\xi) = H(x,\xi)+\tau,
\end{equation}
generating the (extended) Hamiltonian flow
\begin{equation}
\label{extflow}
\Phi^{\rm ext}_\sigma(t,\tau,x,\xi) = \bigl(t+\sigma,\tau,\Phi_\sigma(x,\xi)\bigr).
\end{equation}
Standard FIOs require canonical relations. In the present context this
is a twisted version of the graph of the extended Hamiltonian flow
\eqref{extflow}, which can be given as
\begin{equation}
\label{q5}
\Lambda = \left\{ (t,\tau,x,\xi,y,-\eta);\ 
  (x,\xi)=\Phi_t(y,\eta), \ 
  H_{\rm ext}(t,\tau,y,\eta)=0 \right\},
\end{equation}
and which is a Lagrangian submanifold of $\rz^2\times\tz\times\tz$ or
$\rz^2\times\rz^2\times\rz^2$, equipped with the symplectic form
\begin{equation}
\label{q6}
\tilde\Omega = \ud t\wedge\ud\tau\oplus\omega\oplus\omega.
\end{equation}
We also need local generating functions for $\Lambda$. These always
exist in a neighbourhood of the initial manifold
\cite[Theorem~5.5]{Grigis:1994}
\begin{equation}
\label{a2}
\Lambda_0 = \left\{ (0,\tau,x,\xi,y,-\eta);\ 
  (x,\xi)=(y,\eta),\ 
  H_{\rm ext}(0,\tau,y,\eta)=0 \right\},
\end{equation}
in the form
\begin{equation}
\label{q8}
\phi(t,x,y,\eta)=S(t,x,\eta)-\eta y,
\end{equation}
where $S$ satisfies the Hamilton-Jacobi equation 
\begin{equation}
\label{q9}
H\bigl(x,\partial_x S(t,x,\eta)\bigr)+\partial_t S(t,x,\eta)=0
\quad \text{with initial condition} \quad
S(x,\eta,0)=x\eta.
\end{equation}
A solution $S$ exists when $|t|$ is sufficiently small, so as no
caustic to occur. This solution defines a map
$\phi:\mathfrak{V}_0\to\rz$, where $\mathfrak{V}_0$ is an open set in
$\rz^4$, such that
\begin{equation}
\label{localLM}
\begin{split}
\bigl\{(t,\partial_t S(t,x,\eta),x,\partial_x S(t,x,\eta),
&\partial_\eta S(t,x,\eta),-\eta);\ (t,x,y,\eta)\in\mathfrak{V}_0, \\
&(x,\xi)=\Phi^t(y,\eta),\ H_{\rm ext}(t,\tau,y,\eta)=0\bigr\},
\end{split}
\end{equation}
is a neighbourhood of $\Lambda_0$ in $\Lambda$. The above is
formulated on the covering phase space. Working on the torus would
result in the solution $S$ of the Hamilton-Jacobi equation
$\eqref{q9}$ to be discontinuous.  In order to avoid this problem we
shall use generating functions on the universal cover only.

Representations analogous to \eqref{localLM} can be given at any point of
$\Lambda$ where the map
\begin{equation}
\label{q10}
(x,\xi,y,-\eta)\mapsto (x,\eta), \quad\text{with}\quad \Phi_t(y,\eta)=(x,\xi),
\end{equation} 
is locally surjective. On an open set
$\mathfrak{V}_\alpha\subset\rz^4$ one then finds a function
\begin{equation}
\label{q8a}
\phi_\alpha(t,x,y,\eta)=S_\alpha(t,x,\eta)-\eta y,
\end{equation}
such that the analogue of \eqref{localLM} parametrises another piece
of $\Lambda$.  At points where local surjectivity of \eqref{q10} is
violated but where instead the map
\begin{equation}
\label{q11}
(x,\xi,y,-\eta)\mapsto (\xi,\eta), \quad \quad\text{with}\quad 
\Phi_t(y,\eta)=(x,\xi),
\end{equation}
is locally surjective we may use an alternative generating
function. Slightly abusing our notation, we introduce
$\phi_\alpha:{\mathfrak{V}}_\alpha\to\rz$, where now
$\mathfrak{V}_\alpha$ is an open subset of $\rz^5$, of the form
\begin{equation}
\label{q12}
\phi_\alpha(t,x,y,\xi,\eta)=x\xi-y\eta-S_\alpha(t,\xi,\eta).
\end{equation}
This gives a parametrisation
\begin{equation}
\label{q13}
\begin{split}
\bigl\{(t,\partial_t S_\alpha(t,\xi,\eta),\partial_\xi S_\alpha(t,\xi,\eta),\xi,
&-\partial_\eta S_\alpha(t,\xi,\eta),-\eta);\ (t,x,y,\xi,\eta)\in
  \mathfrak{V}_\alpha, \\
&\qquad(x,\xi)=\Phi^t(y,\eta),\ H_{\rm ext}(t,\tau,y,\eta)=0\bigr\},
\end{split}
\end{equation}
of yet another piece of $\Lambda$. We denote by $\Lambda_\alpha$ the
subset of $\Lambda$ generated by the pair
$\{\phi_\alpha,\mathfrak{V}_\alpha\}$, no matter whether it is
parametrised in the sense of \eqref{localLM} or \eqref{q13}. In fact,
all of $\Lambda$ can be covered by local pieces of the above forms.
This can be proven in analogy to \cite[Lemma~10.5]{Zworski:2012}; in
the case \eqref{q13} one only has to replace the one-form in 
\cite[p.~230]{Zworski:2012} by
\begin{equation}
\label{q195}
\nu:=\tau\ud t - x\ud\xi - y\ud\eta.
\end{equation}
Now let $\{\Psi_\alpha\}_{\alpha}$ be a partition of unity subordinate 
to the open cover $\{\Lambda_\alpha\}_{\alpha}$. We then call
\begin{equation}
\label{q18}
\left\{\Gamma_\alpha\right\}_\alpha,
\quad\text{where}\quad
\Gamma_\alpha = \{\phi_\alpha,\mathfrak{V}_\alpha,\Psi_\alpha,\Lambda_\alpha\},
\end{equation}
a generating set of $\Lambda$.

We remark that if $\gamma(\sigma)=(\sigma,\Phi_\sigma(y,\eta))$ is a
path in $\rz\times\rz^2$ connecting $(0,y,\eta)$ and
$(t,\Phi_t(y,\eta))$ corresponding to $\lambda\in\Lambda$ then its action 
\begin{equation}
\label{q14}
W(\lambda) := \int_{\gamma}\theta + t\tau
\end{equation}
generates $\Lambda$ using suitable local coordinates
\cite[Eq.~(16)]{Meinrenken:1992}.

The pairs $\{\phi_\alpha,\mathfrak{V}_\alpha\}_\alpha$ generate a
distinguished complex line bundle, the Maslov bundle, over
$\Lambda$. This bundle is defined by the transition functions,
\begin{equation}
\label{q15}
\kappa_{\alpha\beta}:=\ue^{\frac{\ui\pi}{2}\sigma_{\alpha\beta}},
\end{equation} 
where
\begin{equation}
\label{q16}
\sigma_{\alpha\beta}(\lambda) :=
\frac{1}{2}\bigl(\sgn\Hess\phi_\alpha(\lambda)-\sgn\Hess\phi_\beta(\lambda)\bigr),
\quad \lambda\in\Lambda_{\alpha}\cap\Lambda_{\beta},
\end{equation}
see \cite[Eq.~(4)]{Meinrenken:1992} and
\cite[Eq.~(21.6.18)]{Hoermander:1985}. Note that here and throughout
the article Hessians are defined only with respect to variables on the
original phase space but not with respect to variables on the extended
phase space. Now assume $\gamma$ to be a continuous path in
$\Lambda$ defined on some compact parameter interval $I$. We choose a
finite partition $t_0,\ldots,t_M$ of $I$ such that
$\gamma([t_{i-1},t_i])\subset\Lambda_{\alpha_i}$ for some
$\alpha_i$. Then we define the Maslov index of $\gamma$ as
\begin{equation}
\label{q17}
\mu(\gamma) := \sum_{i=1}^{M-1}\sigma_{\alpha_{i}\alpha_{i+1}}(\gamma(t_i)),
\end{equation}
compare \cite[p.~288]{Meinrenken:1992}. 

The constructions leading to \eqref{q17} and \eqref{q18} were
performed in the covering space setting. However, they are also valid
in the torus setting, since the constructions are essentially
local. In particular, choosing $\mathfrak{V}_\alpha$, or rather
$\Lambda_\alpha$, small enough and fixing $\alpha$ one can uniquely
identify the coordinates in the torus setting with those in the
covering space setting. If $(x,\xi)$ are coordinates of a point in
$\rz^2$, we shall denote coordinates of its projection to $\tz$ by
$(\check{x}^\alpha,\check\xi^\alpha)$. Conversely, if $(x,\xi)$ are
coordinates of a point on the torus we shall use
$(\hat{x}^\alpha,\hat\xi^\alpha)$ to denote coordinates for this point
lifted to $\rz^2$ in the framework described above. We shall omit the
dependence on $\alpha$ if it is clear from the context.

In order to achieve the above properties we choose a generating set 
\eqref{q18} where the sets $\mathfrak{V}_\alpha$ are small enough to satisfy 
the following conditions: Let $\{z_1,\dots ,z_l\}$ be a subset of the 
coordinates $\{x,y,\xi,\eta\}$ and define $\mathfrak{V}_\alpha^{z_1,\dots,z_l}$
to be the image of $\mathfrak{V}_\alpha$ under the projection 
$(x,y,\xi,\eta)\mapsto(z_1,\dots ,z_l)$. Then we require
\begin{equation}
\label{q26}
\diam\mathfrak{V}_\alpha^{x},\diam\mathfrak{V}_\alpha^{y}<\frac{\ell_x}{2} 
\quad\text{and}\quad \diam\mathfrak{V}_\alpha^{\xi},
\diam\mathfrak{V}_\alpha^{\eta}<\frac{\ell_\xi}{2}
\end{equation}
for every $t\in[0,T]$ with arbitrary but fixed $T$. Here $\diam M$ denotes
the diameter of $M\subset\rz^d$ measured in the euclidean metric.

We now develop a generating set for $\Lambda$ that is adapted to
the construction of discrete scFIOs which is carried out
in Sec.~\ref{q166}. We start with an arbitrary subordinate partition
of unity $\left\{\psi_{\rho},V_{\rho}\right\}_\rho$ of
$\tz$. Translating this partition with $\ell_x\gz\oplus \ell_\xi\gz$
then gives a partition of unity of $\rz^2$. Another related
subordinate partition of unity, $\{\psi'_{\rho,t},V'_{\rho,t}\}_\rho$,
is given by
\begin{equation}
\label{q20}
V'_{\rho,t} = \left\{ (x,\xi),\ \Phi_t(y,\eta)=(x,\xi),\ 
                      (y,-\eta)\in V_{\rho}\right\}
\quad\text{and}\quad 
\psi'_{\rho,t}(x,\xi) = \psi_{\rho}(y,-\eta).
\end{equation}
Let further $\{\kappa_{\beta},I_\beta\}_{\beta}$ be a subordinate
partition of unity of the interval $[0,T]$ and set
$\alpha=(\rho,\beta)$, as well as
\begin{equation}
\label{q21}
\Lambda_\alpha = \left\{(t,\tau,x,\xi,y-\eta)\in\rz^2\times V'_{\rho,t}\times
V_{\rho};\ (t,\tau,x,\xi,y,-\eta)\in\Lambda\right\},
\end{equation}
and 
\begin{equation}
\label{q22}
\Psi_\alpha (t,\tau,x,\xi,y,-\eta) = 
\kappa_\beta(t)\psi'_{\rho,t}(x,\xi)\psi_\rho(y,-\eta).
\end{equation} 
Then $\{\Psi_\alpha,\Lambda_\alpha\}_\alpha$ is a subordinate
partition of unity of $\Lambda$.

The above constructions also work on the covering space. There,
translations by $\ell_x\gz\oplus \ell_\xi\gz$ introduce an equivalence
relation on the collection of sets $\Lambda_{\alpha}$. This generates
equivalence classes $\Lambda_{\widetilde{\alpha}}$,
$\mathfrak{V}_{\widetilde{\alpha}}$ and
$\Psi_{\widetilde{\alpha}}$. In short, we also say that
$\alpha,\alpha'$ are in the equivalence class $\widetilde{\alpha}$.
\begin{lemma}
\label{q168}
Let ${\Lambda}_{\alpha},{\Lambda}_{\alpha'}$ be in the same equivalence class
$\Lambda_{\widetilde{\alpha}}$. Then 
\begin{equation}
\label{q172}
\check{\Lambda}_{\alpha'}=\check{\Lambda}_{\alpha}
\end{equation}
and the converse is also true. Moreover, assume that
$\check{\lambda}_{\alpha}=\check{\lambda}_{\alpha'}$, where
$\lambda_{\alpha'}\in\Lambda_{\alpha'}$ and
$\lambda_{\alpha}\in\Lambda_{\alpha}$. Then
\begin{equation}
\label{171}
W(\lambda_\alpha)=W(\lambda_{\alpha'})
\end{equation}
If, in addition, $\mathfrak{V}_{\alpha'},\mathfrak{V}_\alpha$ are in
the same equivalence class $\mathfrak{V}_{\widetilde{\alpha}}$, then
\begin{equation}
\label{170}
\Hess\phi_{\alpha}(\lambda_{\alpha})=\Hess\phi_{\alpha'}(\lambda_{\alpha'}).
\end{equation}
\end{lemma}
\begin{proof}
Since the Hamiltonian is periodic the relation \eqref{171} follows 
for all $\lambda_\alpha$ and $\lambda_{\alpha'}$ in the set \eqref{q172}. 
This implies that $\phi_\alpha$ and $\phi_{\alpha'}$ can be chosen of the 
same type and differ only by linear terms. Since the latter do not 
contribute to the Hessian this proves \eqref{170}.
\end{proof}
A consequence for the Maslov bundle with transition functions \eqref{q15}
follows immediately.
\begin{cor}
\label{173}
With the assumptions of Lemma~\ref{q168}, as well as
$\alpha,\alpha'\in\widetilde{\alpha}$ and
$\beta,\beta'\in\widetilde{\beta}$, it follows that
\begin{equation}
\label{174}
\kappa_{\alpha\beta}=\kappa_{\alpha'\beta'},
\end{equation}
and
\begin{equation}
\label{169}
\partial_x\xi_\alpha(\lambda_\alpha) = \partial_x\xi_{\alpha'}(\lambda_{\alpha'}),
\quad\text{and}\quad 
\partial_\xi x_\alpha(\lambda_\alpha) = \partial_\xi x_{\alpha'}(\lambda_{\alpha'}).
\end{equation}
\end{cor}
Another useful result is the following.
\begin{lemma}
\label{q190}
Assume the restrictions \eqref{q26} on the size of the sets
$\mathfrak{V}_\alpha$. If the projections $\check{\Lambda}_{\alpha'}$
and $\check{\Lambda}_{\alpha}$ have non-empty intersection then there
exists a unique $\gamma\in \ell_x\gz\oplus \ell_\xi\gz$ such that
\begin{equation}
\label{q192}
\gamma\bigl(\Lambda_\alpha\bigr)\cap\Lambda_{\alpha'}\neq\emptyset,
\end{equation}
where $\gamma(\Lambda_\alpha)$ is the $\gamma$-translate of $\Lambda_\alpha$.
The converse is also true.
\end{lemma}
\begin{proof}
Two points $\lambda\in\Lambda_\alpha$ and $\lambda'\in\Lambda_{\alpha'}$ 
satisfying $\check{\lambda}^{\alpha}=\check{\lambda}'^{\alpha'}$ can only 
differ by a torus translation $\gamma\in \ell_x\gz\oplus \ell_\xi\gz$. 
Suppose that there exists $\gamma'\neq\gamma$ in $\ell_x\gz\oplus \ell_\xi\gz$ 
as well as $\lambda,\tilde{\lambda}\in\Lambda_{\alpha}$ such that both 
$\gamma(\lambda)$ and $\gamma'(\tilde{\lambda})$ are in $\in\Lambda_{\alpha'}$. 
The distance in position and momentum coordinates, respectively, then is 
smaller than $\ell_x/2$ and $\ell_\xi/2$. As a consequence, the distance 
between $\gamma(\lambda)$ and $\gamma'(\tilde{\lambda})$ in the position 
and momentum coordinates contradicts the assumptions \eqref{q26} since 
$\gamma'\neq\gamma$.
\end{proof}
For the purpose of constructing discrete scFIOs we select exactly 
one representative $\Lambda_{\boldsymbol{\alpha}}$ in every equivalence class 
$\Lambda_{\widetilde{\alpha}}$ and denote the resulting generating set by 
\begin{equation}
\label{q196}
\left\{\Gamma_{\boldsymbol{\alpha}}\right\}_{\boldsymbol{\alpha}},\quad
\Gamma_{\boldsymbol{\alpha}} = \left\{\phi_{\boldsymbol{\alpha}},
\mathfrak{V}_{\boldsymbol{\alpha}},\Psi_{\boldsymbol{\alpha}},\Lambda_{\boldsymbol{\alpha}}
\right\}.
\end{equation}
We remark that
$\{\Psi_{\boldsymbol{\alpha}},\Lambda_{\boldsymbol{\alpha}}\}_{\boldsymbol{\alpha}}$
is not a partition of unity for the covering space, but we still have
\begin{equation}
\label{q197}
\sum_{\boldsymbol{\alpha}}\Psi_{\boldsymbol{\alpha}}
(\lambda_{\boldsymbol{\alpha}})=1,
\end{equation}
where the  points $\lambda_{\boldsymbol{\alpha}}$ can be different for the 
different terms in the sum, but their projections 
$\check{\lambda}_{\boldsymbol{\alpha}}$ to the torus are identical.


%
\section{Semiclassical approximation of time evolution}
\label{q166}
Before constructing a semiclassical approximation $U_{\rm scl}(t)$ of
the unitary time evolution $U(t)=\ue^{-\frac{\ui}{\hbar}\op_N(H)t}$ as
in \eqref{scevol}, we want to clarify in what sense such an
approximation is to be understood. Since only the trace of $U(t)$
enters the trace formula \eqref{TFidea}, we need to estimate the
difference $\tr(U_{\rm scl}(t)-U(t))$ in a suitable way. In the
following we show that this difference can be estimated using the
integrated Hilbert-Schmidt norm of an error term.
\begin{lemma}
\label{q50}
Let $T>0$ and let $U_{\rm scl}(t)$ be differentiable for $t\in (0,T)$,
such that
\begin{equation}
\label{q51}
\int_0^T\left\|\bigl(\ui\hbar\partial_t - 
\op_N(H)\bigr)U_{\rm scl}(t)\right\|_{\HS}^2\ud t = O\bigl(\hbar^{\beta}\bigr),
\end{equation}
for some $\beta>0$. Assume that the initial data 
satisfies 
\begin{equation}
\label{q52}
\left\|U_{\rm scl}(0)-\eins\right\|_{\HS}=O(\hbar^{\infty}).
\end{equation}
Then,
\begin{equation}
\label{q53}
\left|\tr\left(U_{\rm scl}(t)-U(t)\right)\right| =
O\bigl(\hbar^{\frac{\beta-3}{2}}\bigr).
\end{equation}
\end{lemma}
\begin{proof}
Let $t\in (0,T)$ and define
\begin{equation}
\label{q54}
R(t)=\bigl(\ui\hbar\partial_t-\op_N(H)\bigr)U_{\rm scl}(t)
\quad\text{and}\quad
Z=U_{\rm scl}(0)-\eins.
\end{equation}
We then have 
\begin{equation}
\label{q56}
\partial_t\bigl(U(-t)U_{\rm scl}(t)\bigr) = -\frac{\ui}{\hbar}U(-t)R(t).
\end{equation}
Hence,
\begin{equation}
\label{q57}
\begin{split}
\left|\tr\bigl(U_{\rm scl}(t)-U(t)\bigr)\right|
  &=\left|\tr\bigl[U(t)\bigl(U(-t)U_{\rm scl}(t)-\eins\bigr)\bigr]\right|\\
  &=\Bigl|\tr\Bigl[U(t)\Bigl(Z+\int_0^t\partial_\sigma\bigl(U(-\sigma)
     U_{\rm scl}(\sigma)\bigr)\ud\sigma\Bigr)\Bigr]\Bigr|\\
  &\leq\tr\bigl|U(t)Z\bigr|+\frac{1}{\hbar}\int_0^t\tr\bigl|U(t-\sigma)R(\sigma)
     \bigr|\,\ud\sigma\\
  &\leq \bigl\|U(t)\bigr\|_{\HS} \, \bigl\|Z\bigr\|_{\HS}
  + \frac{1}{\hbar} \int_0^t \bigl\|U(t-\sigma)\bigr\|_{\HS} \, 
    \bigl\|R(\sigma)\bigr\|_{\HS} \, \ud\sigma,
\end{split}
\end{equation}
where in the last step we have used that trace norm and
Hilbert-Schmidt norm satisfy $\tr|AB| \leq \|A\|_{\HS} \|B\|_{\HS}$,
see e.g.\ \cite[App.~to~IX.4, Prop.~5]{Reed:1975}. Recalling that
$\|Z\|_{\HS}=O(\hbar^\infty)$, keeping in mind that
$\|U(t)\|_{\HS}=\sqrt{N}=O(\hbar^{-1/2})$ and using
\begin{equation}
  \int_0^t \bigl\|R(\sigma)\bigr\|_{\HS} \, \ud\sigma
  \leq \sqrt{t\int_0^t \bigl\|R(\sigma)\bigr\|_{\HS}^2 \, \ud\sigma}
  = O(\hbar^{\beta/2})
\end{equation}
establishes the desired estimate since $t\in(0,T)$.
\end{proof}
The construction of a semiclassical approximation $U_{\rm scl}(t)$ involves
a quantisation of the Lagrangian manifold $\Lambda$ \eqref{q5}. To this end we
have to introduce a variant of scFIOs \cite{Meinrenken:1994} suitable for 
the present context.
\begin{defn}
\label{q58}
Let $\Lambda$ be the Lagrangian submanifold of $\rz^2\times\tz\times\tz$
defined in \eqref{q5} with generating set \eqref{q196}. An operator 
$U_{\rm scl}(t)\in\Mat_N(\kz)$, $t\in(0,T)$, is then said to be a discrete 
scFIO associated with the Lagrangian manifold $\Lambda$, if
\begin{equation}
\label{q59}
U_\mathrm{scl}(t)=\sum_{\boldsymbol{\alpha}}U_{\boldsymbol{\alpha}}(t),
\end{equation}
such that every $U_{\boldsymbol{\alpha}}$ has an asymptotic expansion in the Borel sense 
\cite[Section 4.4.2.]{Zworski:2012},
\begin{equation} 
\label{q60}
U_{\boldsymbol{\alpha}}(t) \sim \sum_{k\in\nz_0}\left(\frac{\hbar}{\ui}\right)^{k}
U_{\boldsymbol{\alpha},k}(t),
\end{equation}
where
\begin{equation}
\label{q167}
U_{\boldsymbol{\alpha},k}(t)_{mn} = \left(\frac{1}{2\pi\hbar}\right)^{\frac{J+1}{2}}
\frac{\ell_x}{N}\int_{\rz^J}a_{\boldsymbol{\alpha},k}(t,\hat{x}_m,\hat{y}_n,
\boldsymbol{\theta})\,\ue^{\frac{\ui}{\hbar}\phi_{\boldsymbol{\alpha}}
(t,\hat{x}_m,\hat{y}_n,\boldsymbol{\theta})}\ \ud\boldsymbol{\theta}.
\end{equation}
Here $\phi_{\boldsymbol{\alpha}}\in
C^\infty(\rz\times\rz^2\times\rz^J)$ are generating functions for the
local pieces $\Lambda_{\boldsymbol{\alpha}}$ of $\Lambda$, and
$a_{\boldsymbol{\alpha}}\in C_0^\infty(\rz\times\rz^2\times\rz^J)$.
\end{defn}
As the only contribution to $U_{\boldsymbol{\alpha}}(t)$ that exceeds
$O(\hbar^\infty)$ in magnitude derives from the Lagrangian manifold 
$\Lambda_{\boldsymbol{\alpha}}$, the integral in \eqref{q167} can be converted 
into an integral over either $\mathfrak{V}_{\boldsymbol{\alpha}}^\eta$, when
$J=1$, or over $\mathfrak{V}_{\boldsymbol{\alpha}}^{\xi,\eta}$, when $J=2$. In the 
first case the integration variable is $\theta=\eta$ and we denote the 
generating function by $\phi^I(t,x,y,\eta)=S^I(t,x,\eta)-y\eta$. In the 
second case we have that $\boldsymbol{\theta}=(\xi,\eta)$ and denote the 
generating function by $\phi^{II}(t,x,y,\xi,\eta)=x\xi-y\eta-S^{II}(t,\xi,\eta)$. 
Furthermore, we choose $\Lambda_{\boldsymbol{\alpha}}$ small enough such that if
$\phi_{\boldsymbol{\alpha}}$ and $\phi_{\boldsymbol{\alpha}'}$ are not of the same 
form $\phi^I$ or $\phi^{II}$, $\lambda\in\check{\Lambda}_{\boldsymbol{\alpha}}
\cap\check{\Lambda}_{\boldsymbol{\alpha}'}$ implies that all derivatives
$\partial_x\xi\bigl(\hat{\lambda}^{\boldsymbol{\alpha}'}\bigr)$,
$\partial_x\xi\bigl(\hat{\lambda}^{\boldsymbol{\alpha}}\bigr)$,
$\partial_\xi x\bigl(\hat{\lambda}^{\boldsymbol{\alpha}'}\bigr)$ and
$\partial_\xi x\bigl(\hat{\lambda}^{\boldsymbol{\alpha}}\bigr)$ do not
vanish.

In order to determine the discrete scFIO that satisfies the estimate
\eqref{q51} to a sufficient order in $\hbar$, we choose a partition
$\{\phi_{\boldsymbol{\alpha}},\mathfrak{V}_{\boldsymbol{\alpha}}\}_{\boldsymbol{\alpha}}$, insert
\eqref{q59} into the left-hand side of \eqref{q51} and define
\begin{equation}
\label{q54a}
R_{\boldsymbol{\alpha},k}(t)=\bigl(\ui\hbar\partial_t-\op_N(H)\bigr)
U_{\boldsymbol{\alpha},k}(t)
\end{equation}
for each coordinate patch $\boldsymbol{\alpha}$ and order
$k\in\nz_0$. Together with \eqref{q54} this implies
\begin{equation}
\label{q155}
R(t)=\sum_{\boldsymbol{\alpha},k}\left(\frac{\hbar}{\ui}\right)^{k}
R_{\boldsymbol{\alpha},k}(t).
\end{equation} 
With Lemma~\ref{q50} in mind we then have to estimate
\begin{equation}
\label{q156}
\tr\bigl(R(t)^\ast R(t)\bigr)=\sum_{\boldsymbol{\alpha}',k',{\boldsymbol{\alpha}},k}
\left(\frac{\hbar}{\ui}\right)^{k+k'}\tr\bigl(R_{\boldsymbol{\alpha}',k'}(t)^\ast 
R_{\boldsymbol{\alpha},k}(t)\bigr)
\end{equation}
semiclassically. This will provide conditions to be satisified by the 
functions $a_{\boldsymbol{\alpha},k}$ in order for $R(t)$ to vanish to any 
finite order in $\hbar$. 

As a first step we now determine representations of
$R_{\boldsymbol{\alpha},k}(t)$ for the two cases \eqref{q8} and
\eqref{q12} separately.
\subsection{The case $\phi^I(t,x,y,\eta)=S^I(t,x,\eta)-y\eta$}
\label{q66}
In this case ${\Lambda}_{\boldsymbol{\alpha}}$ can locally be
parametrised as in \eqref{localLM} using coordinates $(t,x,\eta)$. We
define a function $r_{\boldsymbol{\alpha}}$ through a Taylor
expansion,
\begin{equation}
\label{q67}
S^I_{\boldsymbol{\alpha}}(t,x,\xi) = S^I_{\boldsymbol{\alpha}}(t,y,\xi) +
\partial_x S^I_{\boldsymbol{\alpha}}(t,y,\xi)\,(x-y) + 
r_{\boldsymbol{\alpha}}(t,y,x-y,\xi),
\end{equation}
such that
\begin{equation}
\label{q68}
r_{\boldsymbol{\alpha}}(t,y,0,\xi) = \partial_x r_{\boldsymbol{\alpha}}(t,y,0,\xi)=0.
\end{equation}
We also introduce the abbreviations
\begin{equation}
\label{q70}
\xi_{t,x,\eta} := \partial_xS^I_{\boldsymbol{\alpha}}(t,x,\eta)\quad\text{and}\quad 
y_{t,x,\eta} := \partial_\eta S^I_{\boldsymbol{\alpha}}(t,x,\eta),
\end{equation}
and define the set
\begin{equation}
\label{q70a}
\mathfrak{V}^I_{\boldsymbol{\alpha}} := \left\{(t,x,y,\eta);\ 
(t,\tau,x,y_{t,x,\eta},\xi_{t,x,\eta},-\eta) \in \Lambda_{\boldsymbol{\alpha}},
|y-y_{t,x,\eta}|<\epsilon\right\}.
\end{equation}
The following result determines the application of
$\big(\ui\hbar\partial_t -\op_N(H)\big)$ to an expression of the form
\eqref{q167} with $J=1$.
\begin{lemma}
\label{q71}
Assume that $a^I_{\boldsymbol{\alpha},k}$ is smooth in
$(t,x,y,\eta)\in\mathfrak{V}_{\boldsymbol{\alpha}}$. Then there exists
a smooth function $b_{\boldsymbol{\alpha},k,\hbar}$ with an asymptotic
expansion in $\hbar$ that is uniform in $\mathfrak{V}^I_{\boldsymbol{\alpha}}$,
\begin{equation}
\label{q72}
b_{\boldsymbol{\alpha},k,\hbar}\sim\sum_{n=0}^\infty\hbar^n b_{\boldsymbol{\alpha},k,n},
\end{equation} 
such that the quantity \eqref{q54a} takes the form
\begin{equation}
\label{q72a}
R_{\boldsymbol{\alpha},k}(t)_{mn} = \frac{1}{2\pi\hbar}\frac{\ell_x}{N}
\int_{\mathfrak{V}_{\boldsymbol{\alpha}}^\eta}c^I_{\boldsymbol{\alpha},k}(t,\hat{x}_m,\hat{y}_n,\eta)\,
\ue^{\frac{\ui}{\hbar}(S^I_{\boldsymbol{\alpha}}(t,\hat{x}_m,\eta)-\hat{y}_m\eta)}\ \ud\eta.
\end{equation}
Here,
\begin{equation}
\label{q73}
\begin{split}
c^I_{\boldsymbol{\alpha},k}(t,x,y,\eta)
  &:=\ui\hbar(\partial_t a^I_{\boldsymbol{\alpha},k})(t,x,y,\eta) 
      -\frac{\hbar}{2\ui}(\partial_x\partial_{\xi}H)(x,\xi_{t,x,\eta})
      a^I_{\boldsymbol{\alpha},k}(t,x,y,\eta)\\
  &\quad-\frac{\hbar}{\ui}(\partial_{\xi}H)(x,\xi_{t,x,\eta})
     (\partial_x a^I_{\boldsymbol{\alpha},k})(t,x,y,\eta)
     -\left(\frac{\hbar}{\ui}\right)^2 b_{\boldsymbol{\alpha},k,\hbar}(t,x,y,\eta).
\end{split}
\end{equation}
\end{lemma}
\begin{proof}
  The proof follows the standard WKB-type strategy (see e.g.
  \cite[p.\,105]{Duistermaat:1973}). We use the representation
  \eqref{q167} with $J=1$, and the variable $\boldsymbol{\theta}$ is
  denoted as $\eta$. The first step in evaluating \eqref{q54} is to
  use Lemma~\ref{q39} for the application of $\op_N(H)$ to
  $U_{\boldsymbol{\alpha},k}$ with the matrix elements
  $({U_{\boldsymbol{\alpha},k}})_{ln}$, where $n$ is fixed, in place
  of $\psi_{l}$, as well as a summation over $l$ with $|l-m|\leq
  N^{\frac{1}{M}}$ in \eqref{q40}.  For the phase function in
  \eqref{q167} we use \eqref{q67}, and then perform a Taylor expansion
  \eqref{Taylorexp} of
\begin{equation}
\label{q75}
H\left(\frac{x_m+x_l}{2},\xi\right)
a^I_{\boldsymbol{\alpha},k}(t,\hat{x}_l,\hat{y}_n,\eta) \,
\ue^{\frac{\ui}{\hbar}r_{\boldsymbol{\alpha}}(t,\hat{x}_l,\hat{x}_m-\hat{x}_l,\eta)}
\end{equation}
in $x_m-x_l$. We integrate by parts, use \eqref{q46} and the
periodicity of $H$, and obtain an asymptotic expansion in powers of 
$\hbar$ multiplied by $\ue^{\frac{\ui}{\hbar}(\xi-\partial_x S^I(t,\hat{x}_m,\eta))(x_m-x_l)}$.  
If the sum over $l$ extended over all $l\in\gz$, it would be a Fourier 
series, with the $\xi$-integrals as Fourier coefficients. The latter, 
however, are of size $O(|m-l|^{-\infty})$. Hence, while adding an error of 
size $O(\hbar^\infty)$ one can indeed extend the restricted sum over $l$ to
all $l\in\gz$. The resulting Fourier series can be evaluated, thereby
replacing the variable $\xi$ with $\xi_{t,\hat{x}_m,\eta}$.

When applying $\ui\hbar\partial_t$ to $U_{\boldsymbol{\alpha},k}$ we use the 
Hamilton-Jacobi equation \eqref{q9} and evaluate the coefficients of the 
lowest powers in $\hbar$. This gives \eqref{q73} by noting that 
$\supp a^I_{\boldsymbol{\alpha},k}\subset\mathfrak{V}_{\boldsymbol{\alpha}}$.
\end{proof}
For later reference we note that the quantity \eqref{q73}, viewed as
the coefficient function of a half-density, can be identified as a
Lie-derivative along the Hamiltonian vector field generated by the
(extended) classical Hamiltonian,
\begin{equation}
\label{q100}
c^I_{\boldsymbol{\alpha},k}(\lambda)\ud\lambda^{\frac{1}{2}} = \ui\hbar 
\mathcal{L}_{X_{H_{\rm ext}}}\left(a^I_{\boldsymbol{\alpha},k}(\lambda)\ud\lambda^{\frac{1}{2}}
\right)+O(\hbar^2), 
\end{equation}
see \cite[(10.2.22),\,(10.2.27)]{Zworski:2012}.
\subsection{The case $\phi^{II}(t,x,y,\xi,\eta)=x\xi-y\eta-S^{II}(t,\xi,\eta)$}
\label{q76}
In the remaining case, a parametrisation of $\Lambda_{\boldsymbol{\alpha}}$ 
in the sense of $\eqref{q11}$ can be achieved with coordinates 
$(t,\xi,\eta)$. We also introduce the abbreviations
\begin{equation}
\label{q77}
x_{t,\xi,\eta}:=\partial_\xi S^{II}_{\boldsymbol{\alpha}}(t,\xi,\eta)\quad\text{and}\quad
y_{t,\xi,\eta}:=-\partial_\eta S^{II}_{\boldsymbol{\alpha}}(t,\xi,\eta),
\end{equation}
and choose a smooth function $\kappa_{\epsilon}:\rz\to[0,1]$ satisfying
\begin{equation}
\label{q78}
\kappa_{\epsilon}\left(x\right)=
\begin{cases}
0, & |x|>\frac{2\epsilon}{3},\\
1, & |x|<\frac{\epsilon}{3},
\end{cases}
\end{equation}
that is used to localise around initial and final points \eqref{q77} of classical 
trajectories. We now assume that $A_{\boldsymbol{\alpha},k}$ is a function on 
$\Lambda_{\boldsymbol{\alpha}}$, therefore depending on the variables $(t,\xi,\eta)$,
and define 
\begin{equation}
\label{q79}
a^{II}_{\boldsymbol{\alpha},k}(t,\hat{x}_m,\hat{y}_n,\xi,\eta) = 
\kappa_{\epsilon}(\hat{x}_m-x_{t,\xi,\eta})\kappa_{\epsilon}(\hat{y}_n-y_{t,\xi,\eta})
A_{\boldsymbol{\alpha},k}(t,\xi,\eta).
\end{equation}
This function is localised on the set
\begin{equation}
\label{q91}
\begin{split}
\mathfrak{V}^{II}_{\boldsymbol{\alpha},\epsilon}:=
  &\bigl\{(t,x,y,\xi,\eta); \ (t,\tau,x_{t,\xi,\eta},\xi,y_{t,\xi,\eta},-\eta)\in
     \Lambda_{\boldsymbol{\alpha}},\\
  &\qquad |x-x_{t,\xi,\eta}|<\epsilon, \ |y-y_{t,x,\eta}|<\epsilon\bigr\}.
\end{split}
\end{equation}
Note that here, in contrast to the previous case \eqref{q70a}, we admit
the variables $x,y$ to be from a (small) neighbourhood of 
$\Lambda_{\boldsymbol{\alpha}}$.

In the present case, the application of $\ui\hbar\partial_t -\op_N(H)$ to the
respective expression of the form \eqref{q167} gives the following results.
\begin{lemma}
\label{q80}
Using a function of the form \eqref{q79} in \eqref{q167} one gets
\begin{equation}
\label{q82}
R_{\boldsymbol{\alpha},k}(t)_{mn} = \frac{1}{(2\pi\hbar)^{\frac{3}{2}}}\frac{\ell_x}{N}
\int_{\mathfrak{V}_{\boldsymbol{\alpha}}^{\xi,\eta}}c^{II}_{\boldsymbol{\alpha},k}(t,\hat{x}_m,\hat{y}_n,\xi,\eta)\,
\ue^{\frac{\ui}{\hbar}(\hat{x}_m\xi-\hat{y}_n\eta-S^{II}_{\boldsymbol{\alpha}}(t,\xi,\eta))}\ \ud\xi\ud\eta,
\end{equation}
where
\begin{equation}
\begin{split}
\label{cIIresult}
c^{II}_{\boldsymbol{\alpha},k}(t,x,y,\xi,\eta) 
 &:= \ui\hbar\partial_t a^{II}_{\boldsymbol{\alpha},k}(t,x,y,\xi,\eta) + \frac{\hbar}{\ui}
     (\partial_x H)(x_{t,\xi,\eta},\xi)(\partial_\xi a^{II}_{\boldsymbol{\alpha},k})(t,x,y,\xi,\eta)\\
 &\quad+\frac{\hbar}{2\ui}(\partial_x\partial_\xi H)(x_{t,\xi,\eta},\xi)
     a^{II}_{\boldsymbol{\alpha},k}(t,x,y,\xi,\eta)+O\bigl(\hbar^2\bigr). 
\end{split}
\end{equation}
\end{lemma}
\begin{proof}
Here the representation \eqref{q167} requires $J=2$ and the integration variables
are $\boldsymbol{\theta}=(\xi,\eta)$. The first step in evaluating \eqref{q54} 
consists of using Lemma~\ref{q42} to evaluate the application of $\op_N(H)$ 
to $U_{\boldsymbol{\alpha},k}$. Here, again, the matrix elements 
$({U_{\boldsymbol{\alpha},k}})_{ln}$ replace $\psi_{l}$ and the summation extends over 
$l$ with $|l-m|\leq N^{\frac{1}{M}}$ such that $n$ is fixed. This sum can be
evaluated with Proposition~\ref{a6}. This gives \eqref{q82}, where
\begin{equation}
\label{cIIprelim}
\begin{split}
c^{II}_{\boldsymbol{\alpha},k}(t,\hat{x}_m,\hat{y}_n,\xi,\eta) 
 &= \ui\hbar\partial_t a^{II}_{\boldsymbol{\alpha},\hbar}(t,\hat{x}_m,\hat{y}_n,\xi,\eta) \\
 &\quad+\bigl(H(x_{t,\xi,\eta},\xi)-H(\hat{x}_m,\xi)\bigr)
   a^{II}_{\boldsymbol{\alpha},k}(t,\hat{x}_m,\hat{y}_n,\xi,\eta)\\
 &\quad-\frac{\hbar}{2\ui}(\partial_x\partial_\xi H)(\hat{x}_m,\xi)
   a^{II}_{\boldsymbol{\alpha},k}(t,\hat{x}_m,\hat{y}_n,\xi,\eta)+O\bigl(\hbar^2\bigr). 
\end{split}
\end{equation}
We now employ a Taylor expansion (see also \cite[p.\,23]{Hoermander:1994}), 
\begin{equation}
\label{q116}
\begin{split}
\bigl(H(\hat{x}_m,\xi)-H(x_{t,\xi,\eta},\xi)\bigr)
&\ue^{\frac{\ui}{\hbar}\phi^{II}_{\boldsymbol{\alpha}}(t,\hat{x}_m,\hat{y}_n,\xi,\eta)}\\
&=p_{\boldsymbol{\alpha}}(t,\hat{x}_m,\xi,\eta)\bigl(\hat{x}_m-{x_{t,\xi,\eta}}\bigr)
    \ue^{\frac{\ui}{\hbar}\phi^{II}_{\boldsymbol{\alpha}}(t,\hat{x}_m,\hat{y}_n,\xi,\eta)}\\
&=p_{\boldsymbol{\alpha}}(t,\hat{x}_m,\xi,\eta)\frac{\hbar}{\ui}
    \partial_\xi\ue^{\frac{\ui}{\hbar}\phi^{II}_{\boldsymbol{\alpha}}(t,\hat{x}_m,\hat{y}_n,\xi,\eta)},
\end{split}
\end{equation}
use this in \eqref{cIIprelim} and perform an integration by parts in the
variable $\xi$. By \cite[p.\,24]{Hoermander:1994} one sees that this leads 
to \eqref{cIIresult}.
\end{proof}
We remark that, in analogy to \eqref{q100},  on $\Lambda_{\boldsymbol{\alpha}}$,
\begin{equation}
\label{q100a}
c^{II}_{\boldsymbol{\alpha},k}(\lambda)\ud\lambda^{\frac{1}{2}} = \ui\hbar 
\mathcal{L}_{X_{H_{\rm ext}}}\left(A_{\boldsymbol{\alpha},k}(\lambda)\ud\lambda^{\frac{1}{2}}
\right)+O(\hbar^2), 
\end{equation}
see \cite[(10.2.22),\,(10.2.27)]{Zworski:2012}.
\subsection{Some auxiliary results}
With \eqref{q156} in mind we now estimate 
\begin{equation}
\int_0^T\sum_{n,m=0}^{N-1}\overline{R_{\boldsymbol{\alpha}',k'}(t)}_{nm} 
R_{\boldsymbol{\alpha},k}(t)_{nm}\ \ud t 
\end{equation}
for all pairs $(\boldsymbol{\alpha},\boldsymbol{\alpha}')$ and fixed $T>0$.
There are four cases, depending on whether Lemma~\ref{q71} or
Lemma~\ref{q80} applies to $R_{\boldsymbol{\alpha},k}$ and $R_{\boldsymbol{\alpha}',k'}$,
respectively. We use the notations introduced in Sections~\ref{q66} and 
\ref{q76}, in particular those for the various amplitude and phase functions. 

We recall the definition \eqref{q21} of the sets $\Lambda_\alpha$, as well
as the representatives $\Lambda_{\boldsymbol{\alpha}}$ of equivalence classes
$\Lambda_{\widetilde\alpha}$ (with respect to translations by 
$\ell_x\gz\oplus \ell_\xi\gz$) as described around \eqref{q196}. We then define the 
set
\begin{equation}
\label{q203}
\Lambda_{\boldsymbol{\alpha}'\boldsymbol{\alpha}} :=
\left\{\hat{\lambda}^{\boldsymbol{\alpha}'};\ \lambda\in
\check{\Lambda}_{\boldsymbol{\alpha}}\cap\check{\Lambda}_{\boldsymbol{\alpha}'}\right\}.
\end{equation} 
We also use the interval $I_\beta\subset [0,T]$ as defined in 
\eqref{q20}-\eqref{q22}. 

The four cases that occur are covered in the Lemmata~\ref{q93}--\ref{q130},
the first one being:
\begin{lemma}
\label{q93}
Let $R_{\boldsymbol{\alpha},k}$ be of the form \eqref{q72a} and
$R_{\boldsymbol{\alpha}',k'}$ be of the form \eqref{q82}, then
\begin{equation}
\label{q94}
\begin{split}
\frac{1}{(2\pi\hbar)^{\frac{5}{2}}}\left(\frac{\ell_x}{N}\right)^2
&\sum_{m,n=0}^{N-1}\int_{I_\beta}\int_{\mathfrak{V}_{\boldsymbol{\alpha}'}^{\xi',\eta'}}
  \int_{\mathfrak{V}_{\boldsymbol{\alpha}}^\eta}
  \overline{c^{II}_{\boldsymbol{\alpha}',k'}(t,\hat{x}_m,\hat{y}_n,\xi',\eta')}\,
  c^I_{\boldsymbol{\alpha},k}(t,\hat{x}_m,\hat{y}_n,\eta)\\
 &\hspace*{2cm}\ue^{\frac{\ui}{\hbar}(\phi^I_{\boldsymbol{\alpha}}(t,\hat{x}_m,\hat{y}_n,\eta)-
   \phi^{II}_{\boldsymbol{\alpha}'}(t,\hat{x}_m,\hat{y}_n,\xi',\eta'))}\ud\eta\ud\xi'
    \ud\eta'\ud t\\
 &=\frac{1}{2\pi\hbar}\int_{\Lambda_{\boldsymbol{\alpha}'\boldsymbol{\alpha}}}
    g_{\boldsymbol{\alpha}',k',\hbar}
    \ue^{\frac{\ui\pi}{4}\sgn\Hess\phi^{I}_{\boldsymbol{\alpha}}}|\partial_x\xi|^{-\frac{1}{2}}
    e_{\boldsymbol{\alpha},k,\hbar}\ud t\ud\xi'\ud\eta' + O\bigl(\hbar^{\infty}\bigr),
\end{split}
\end{equation}
where $g_{\boldsymbol{\alpha}',k',\hbar}$ and $e_{\boldsymbol{\alpha},k,\hbar}$ are 
coefficients of half-densities on  $\Lambda_{\boldsymbol{\alpha}'\boldsymbol{\alpha}}$ 
possessing asymptotic expansions in $\hbar$ that are uniform in $t\in[0,T]$. 
In particular,
\begin{equation}
\label{q96}
\begin{split}
e_{\boldsymbol{\alpha},k,\hbar}(\lambda)
  &\sim \left(\ui\hbar\mathcal{L}_{X_{H_{\rm ext}}}
     {a}_{\boldsymbol{\alpha},k}(\lambda) +
    \sum_{n\geq 2}\hbar^n e_{\boldsymbol{\alpha},k,n}(\lambda)\right),
\end{split}
\end{equation}
with $\lambda=(t,\xi',\eta')\in\Lambda_{\boldsymbol{\alpha}'\boldsymbol{\alpha}}$. 
\end{lemma}
\begin{proof}
We first apply Lemma~\ref{a1} to the sums over $n$ and $m$ on the left-hand
side of \eqref{q94}. Up to an error term of size $O(\hbar^{\infty})$ this gives 
two additional integrals,
\begin{equation}
\label{q200}
\frac{1}{(2\pi\hbar)^{\frac{5}{2}}}\int_{I_\beta}\int_0^{\ell_x}\int_0^{\ell_x}
\int_{\mathfrak{V}_{\boldsymbol{\alpha}'}^{\xi',\eta'}}\int_{\mathfrak{V}_{\boldsymbol{\alpha}}^{\eta}}
\overline{c^{II}_{\boldsymbol{\alpha}',k'}}\,c^I_{\boldsymbol{\alpha},k}
\ue^{\frac{\ui}{\hbar}(\phi^I_{\boldsymbol{\alpha}}-\phi^{II}_{\boldsymbol{\alpha}'}+sx\ell_\xi-sy\ell_\xi)}
\ud\eta\ud\xi'\ud\eta'\ud x\ud y\ud t
\end{equation}
where $s\in\gz$ is determined by Lemma \ref{q190} and Lemma \ref{a1}. We then 
use the stationary phase theorem to evaluate the integrals over the variables
$(x,y,\eta)$. In order to determine the stationary points we use 
Lemma~\ref{q190} and get
\begin{equation}
\label{q97}
\begin{aligned}
\partial_y (\phi^I_{\boldsymbol{\alpha}}-\phi^{II}_{\boldsymbol{\alpha}'})-s\ell_\xi = 0\quad
&\Leftrightarrow \quad\eta'=\eta+s\ell_\xi\\
\partial_x (\phi^I_{\boldsymbol{\alpha}}-\phi^{II}_{\boldsymbol{\alpha}'})+s\ell_\xi = 0\quad
&\Leftrightarrow\quad\xi'=\partial_x S^I_{\boldsymbol{\alpha}}(t,x,\eta)+s\ell_\xi\\
\partial_\eta (\phi^I_{\boldsymbol{\alpha}}-\phi^{II}_{\boldsymbol{\alpha}'}) = 0\quad
&\Leftrightarrow\quad\hat{y}^{\boldsymbol{\alpha}}=\partial_\eta 
S^I_{\boldsymbol{\alpha}}(t,x,\eta).
\end{aligned}
\end{equation}
With regard to the last equation we remark that the lifts 
$\hat{y}^{\boldsymbol{\alpha}}$ and $\hat{y}^{\boldsymbol{\alpha}'}$ to the covering
space can only differ (locally) by a constant $r\ell_x$, $r\in\gz$.

The equations \eqref{q97} can be solved to find $x(t,\xi',\eta')$ and 
$y(t,\xi',\eta')$. They determine the set of stationary points that can 
either be given in the form 
\begin{equation}
\label{q98c}
\bigl(t,x(t,\xi',\eta')-r\ell_x,y(t,\xi',\eta')-r\ell_x,\eta'-s\ell_\xi\bigr)
\in\mathfrak{V}^I_{\boldsymbol{\alpha}},
\end{equation}
or
\begin{equation}
\label{q99c}
\bigl(t,x(t,\xi',\eta'),y(t,\xi',\eta'),\xi',\eta')\in
\mathfrak{V}^{II}_{\boldsymbol{\alpha}'}.
\end{equation}
Due to
Lemma~\ref{q168} the Hessian of the total phase, as a function of the
variables $(x,y,\eta)$, on the left-hand side of \eqref{q94} is given
by
\begin{equation}
\label{q101}
\Hess_{x,y,\eta}(\phi^I_{\boldsymbol{\alpha}}-\phi^{II}_{\boldsymbol{\alpha}'}) =
\Hess\phi^I_{\boldsymbol{\alpha}}=
\begin{pmatrix}
\partial_x^2S^I_{\boldsymbol{\alpha}} & 0 & \partial_{x,\eta}S^I_{\boldsymbol{\alpha}}\\
0 & 0 & -1\\
\partial_{x,\eta}S^I_{\boldsymbol{\alpha}} & -1 & \partial^2_{\eta}S^I_{\boldsymbol{\alpha}}
\end{pmatrix}.
\end{equation} 
Recall that our definition of Hessians never involves derivatives with
respect to $t$. Hence,
\begin{equation}
\label{q102}
\det\Hess_{x,y,\eta}(\phi^I_{\boldsymbol{\alpha}}-\phi^{II}_{\boldsymbol{\alpha}'}) =
-\partial_x^2S^I_{\boldsymbol{\alpha}} = -\partial_x\xi.
\end{equation}
We recall that we have chosen the partitions small enough so that this
derivative does not vanish. Moreover, 
\begin{equation}
\label{q104}
(\phi^I_{\boldsymbol{\alpha}}-\phi^{II}_{\boldsymbol{\alpha}'})(\lambda) = 
W(\lambda)-W(\lambda) = 0,
\end{equation}
and
\begin{equation}
\label{q105}
\sgn\Hess_{x,y,\eta}(\phi^I_{\boldsymbol{\alpha}}-\phi^{II}_{\boldsymbol{\alpha}'})(\lambda) =
\sgn\Hess\phi^I_{\boldsymbol{\alpha}}(\lambda).
\end{equation}
An application of the stationary phase theorem then gives \eqref{q94}.
The asymptotic expansion \eqref{q96} follows from \eqref{q100}.
\end{proof}
The next case is:
\begin{lemma}
\label{q108}
Let both $R_{\boldsymbol{\alpha},k}$ and $R_{\boldsymbol{\alpha}',k'}$ be of the 
form \eqref{q82}, then
\begin{equation}
\label{q109}
\begin{split}
\frac{1}{(2\pi\hbar)^3}\left(\frac{\ell_x}{N}\right)^2
&\sum_{m,n=0}^{N-1}\int_{I_\beta}\int_{\mathfrak{V}_{\boldsymbol{\alpha}'}^{\xi',\eta'}}
  \int_{\mathfrak{V}_{\boldsymbol{\alpha}}^{\xi\eta}}
  \overline{c^{II}_{\boldsymbol{\alpha}',k'}(t,\hat{x}_m,\hat{y}_n,\xi',\eta')}\,
  c^{II}_{\boldsymbol{\alpha},k}(t,\hat{x}_m,\hat{y}_n,\xi,\eta)\\
 &\hspace*{2cm}\ue^{\frac{\ui}{\hbar}(\phi^{II}_{\boldsymbol{\alpha}}(t,\hat{x}_m,\hat{y}_n,\xi,\eta)-
   \phi^{II}_{\boldsymbol{\alpha}'}(t,\hat{x}_m,\hat{y}_n,\xi',\eta'))}\ud\xi\ud\eta
   \ud\xi'\ud\eta'\ud t\\
 &=\frac{1}{2\pi\hbar}\int_{\Lambda_{\boldsymbol{\alpha}'\boldsymbol{\alpha}}}
    g_{\boldsymbol{\alpha}',k',\hbar}
    \ue^{\frac{\ui\pi}{4}\sgn\Hess\phi^{II}_{\boldsymbol{\alpha}}}e_{\boldsymbol{\alpha},k,\hbar}
    \ud t\ud\xi'\ud\eta' + O\bigl(\hbar^{\infty}\bigr),
\end{split}
\end{equation}
where $g_{\boldsymbol{\alpha}',k',\hbar}$ and $e_{\boldsymbol{\alpha},k,\hbar}$ are 
coefficients of half-densities on  $\Lambda_{\boldsymbol{\alpha}'\boldsymbol{\alpha}}$ 
possessing asymptotic expansions in $\hbar$ that are uniform in $t\in[0,T]$. 
In particular,
\begin{equation}
\label{q111}
\begin{split}
e_{\boldsymbol{\alpha},k,\hbar}(\lambda)
  &\sim \left(\ui\hbar\mathcal{L}_{X_{H_{\rm ext}}}
     {a}_{\boldsymbol{\alpha},k}(\lambda) +
    \sum_{n\geq 2}\hbar^n e_{\boldsymbol{\alpha},k,n}(\lambda)\right),
\end{split}
\end{equation}
with $\lambda=(t,\xi',\eta')\in\Lambda_{\boldsymbol{\alpha}'\boldsymbol{\alpha}}$.
\end{lemma}
\begin{proof}
  We use a similar strategy as in the proof for Lemma \ref{q93}, and
  first perform the sum over $n$ and $m$ by applying
  Lemma~\ref{a1}. With the appropriate modifications, including an
  additional integration over the variable $\xi$ and a pre-factor
  $1/(2\pi\hbar)^3$, this gives a result similar to
  \eqref{q200}. We then apply the stationary phase theorem to the
  integrations over the variables $(x,y,\xi,\eta)$. The conditions
  determining stationary points of the phase function are
\begin{equation}
\label{q112}
\begin{aligned}
\partial_x (\phi^{II}_{\boldsymbol{\alpha}}-\phi^{II}_{\boldsymbol{\alpha}'})+s\ell_\xi = 0\quad
&\Leftrightarrow\quad{\xi}+s\ell_\xi=\xi'\\
\partial_y (\phi^{II}_{\boldsymbol{\alpha}}-\phi^{II}_{\boldsymbol{\alpha}'})-s\ell_\xi = 0\quad
&\Leftrightarrow \quad{\eta}+s\ell_\xi=\eta'\\
\partial_{{\xi}}(\phi^{II}_{\boldsymbol{\alpha}}-\phi^{II}_{\boldsymbol{\alpha}'}) = 0\quad
&\Leftrightarrow\quad \partial_{\xi}S^{II}_{\boldsymbol{\alpha}}(t,\xi,\eta)=
\hat{x}^{\boldsymbol{\alpha}}\\
\partial_{\eta}(\phi^{II}_{\boldsymbol{\alpha}}-\phi^{II}_{\boldsymbol{\alpha}'}) = 0\quad
&\Leftrightarrow\quad -\partial_{\eta}S^{II}_{\boldsymbol{\alpha}}(t,\xi,\eta)=
\hat{y}^{\boldsymbol{\alpha}}.
\end{aligned}
\end{equation}
These equations can be solved for $x(t,\xi',\eta')$ and $y(t,\xi',\eta')$. 
The set of stationary points can be either given as 
\begin{equation}
\label{q114}
\bigl(t,x(t,\xi',\eta')-r\ell_x,y(t,\xi',\eta')-r\ell_x,\xi'-s\ell_\xi,\eta'-s\ell_\xi\bigr)
\in\mathfrak{V}^{II}_{\boldsymbol{\alpha}},
\end{equation}
or as
\begin{equation}
\label{q210}
\bigl(t,x(t,\xi',\eta'),y(t,\xi',\eta'),\xi',\eta'\bigr)\in
\mathfrak{V}^{II}_{\boldsymbol{\alpha}'}.
\end{equation}
The Hessian of the phase function is
\begin{equation}
\label{q119}
\Hess_{x,y,{\xi},{\eta}}(\phi^{II}_{\boldsymbol{\alpha}}-\phi^{II}_{\boldsymbol{\alpha}'})
=\Hess\phi^{II}_{\boldsymbol{\alpha}}=
\begin{pmatrix}
0 & 0 & 1 & 0\\
0 & 0 & 0 & -1\\
1 & 0 & -\partial^2_\xi S^{II}_{\boldsymbol{\alpha}} & -\partial_{{\eta}{\xi}}S^{II}_{\boldsymbol{\alpha}}\\
0 & -1 & -\partial_{{\xi}{\eta}}S^{II}_{\boldsymbol{\alpha}} & -\partial^2_{{\eta}}S^{II}_{\boldsymbol{\alpha}}
\end{pmatrix},
\end{equation}
so that 
\begin{equation}
\label{q120}
\det\Hess_{x,y,\xi,\eta}(\phi^{II}_{\boldsymbol{\alpha}}-\phi^{II}_{\boldsymbol{\alpha}'})=1.
\end{equation}
The stationary phase theorem, together with \eqref{q100a}, now gives
the result.
\end{proof}
The third case is:
\begin{lemma}
\label{q121}
Let both $R_{\boldsymbol{\alpha},k}$ and $R_{\boldsymbol{\alpha}',k'}$ be of the 
form \eqref{q72a}, then
\begin{equation}
\label{q122}
\begin{split}
\frac{1}{(2\pi\hbar)^2}\left(\frac{\ell_x}{N}\right)^2
&\sum_{m,n=0}^{N-1}\int_{I_\beta}\int_{\mathfrak{V}_{\boldsymbol{\alpha}'}^{\eta'}}
  \int_{\mathfrak{V}_{\boldsymbol{\alpha}}^\eta}
  \overline{c^I_{\boldsymbol{\alpha}',k'}(t,\hat{x}_m,\hat{y}_n,\eta')}\,
  c^I_{\boldsymbol{\alpha},k}(t,\hat{x}_m,\hat{y}_n,\eta)\\
 &\hspace*{2cm}\ue^{\frac{\ui}{\hbar}(\phi^I_{\boldsymbol{\alpha}}(t,\hat{x}_m,\hat{y}_n,\eta)-
   \phi^I_{\boldsymbol{\alpha}'}(t,\hat{x}_m,\hat{y}_n,\eta'))}\ud\eta\ud\eta'\ud t\\
 &=\frac{1}{2\pi\hbar}\int_{\Lambda_{\boldsymbol{\alpha}'\boldsymbol{\alpha}}}
    g_{\boldsymbol{\alpha}',k',\hbar}\ue^{\frac{\ui\pi}{4}f_{\boldsymbol{\alpha}}}
    e_{\boldsymbol{\alpha},k,\hbar}\ud t\ud x'\ud\eta' + O\bigl(\hbar^{\infty}\bigr),
\end{split}
\end{equation}
where $g_{\boldsymbol{\alpha}',k',\hbar}$ and $e_{\boldsymbol{\alpha},k,\hbar}$ are 
coefficients of half-densities on  $\Lambda_{\boldsymbol{\alpha}'\boldsymbol{\alpha}}$ 
possessing asymptotic expansions in $\hbar$ that are uniform in $t\in[0,T]$. 
In particular,
\begin{equation}
\label{q96a}
\begin{split}
e_{\boldsymbol{\alpha},k,\hbar}(\lambda)
  &\sim \left(\ui\hbar\mathcal{L}_{X_{H_{\rm ext}}}
     {a}_{\boldsymbol{\alpha},k}(\lambda) +
    \sum_{n\geq 2}\hbar^n e_{\boldsymbol{\alpha},k,n}(\lambda)\right),
\end{split}
\end{equation}
with $\lambda=(t,x',\eta')\in\Lambda_{\boldsymbol{\alpha}'\boldsymbol{\alpha}}$. 
Moreover,
\begin{equation}
\label{q123}
f_{\boldsymbol{\alpha}} := \sgn\Hess_{y,\eta}\phi^I_{\boldsymbol{\alpha}}
=\sgn\begin{pmatrix}
0 & -1\\
-1 & \partial_{\eta}^2 S^I_{\boldsymbol{\alpha}}(\lambda)
\end{pmatrix}.
\end{equation}
\end{lemma}
\begin{proof}
We again employ the same strategy as in the proofs of Lemmata~\ref{q93} and 
\ref{q108}. The sums over $n$ and $m$ are evaluated with the help of Lemma~\ref{a1}
and, with some modifications, this gives a result similar to \eqref{q200}. There is
integration over the variable $\xi'$ and, for convenience, we name the integration
variables introduced by Lemma~\ref{a1} as $x'$ and $y$. Next we apply the stationary 
phase theorem to the integrations over the variables $(y,\eta)$. The conditions 
determining stationary points of the phase function are 
\begin{equation}
\label{q127}
\begin{aligned}
\partial_y(\phi^I_{\boldsymbol{\alpha}}-\phi^I_{\boldsymbol{\alpha}'})-s\ell_\xi = 0\quad
&\Leftrightarrow \quad\eta'=\eta+s\ell_\xi\\
\partial_{\eta}(\phi^I_{\boldsymbol{\alpha}}-\phi^I_{\boldsymbol{\alpha}'}) = 0\quad
&\Leftrightarrow\quad\hat{y}^{\boldsymbol{\alpha}} = 
\partial_{\eta}S^I_{\boldsymbol{\alpha}}(t,x',\eta),
\end{aligned}
\end{equation}
and can be solved for $y(t,x',\eta')$. The stationary points can therefore be given
either in the form
\begin{equation}
\label{q128}
\bigl(t,x-r\ell_x,y(t,x',\eta')-r\ell_x,\eta'-s\ell_\xi\bigr)\in\mathfrak{V}^I_{\boldsymbol{\alpha}}
\end{equation}
or as
\begin{equation}
\label{q204}
\bigl(t,x',y\left(t,x',\eta'\right),\eta'\bigr)\in\mathfrak{V}^I_{\boldsymbol{\alpha}'}.
\end{equation}
The Hessian that is required by the stationary phase theorem is the one given in
\eqref{q123}, hence
\begin{equation}
\label{q129}
\det\Hess_{y,\eta}(\phi^I_{\boldsymbol{\alpha}}-\phi^I_{\boldsymbol{\alpha}'}) = -1.
\end{equation}
\end{proof}
Finally, the last case is:
\begin{lemma}
\label{q130}
Let $R_{\boldsymbol{\alpha},k}$ be of the form \eqref{q82} and
$R_{\boldsymbol{\alpha}',k'}$ be of the form \eqref{q72a}, then
\begin{equation}
\label{q131}
\begin{split}
\frac{1}{(2\pi\hbar)^{\frac{5}{2}}}\left(\frac{\ell_x}{N}\right)^2
&\sum_{m,n=0}^{N-1}\int_{I_\beta}\int_{\mathfrak{V}_{\boldsymbol{\alpha}'}^{\eta'}}
  \int_{\mathfrak{V}_{\boldsymbol{\alpha}}^{\xi,\eta}}
  \overline{c^I_{\boldsymbol{\alpha}',k'}(t,\hat{x}_m,\hat{y}_n,\eta')}\,
  c^{II}_{\boldsymbol{\alpha},k}(t,\hat{x}_m,\hat{y}_n,\xi,\eta)\\
 &\hspace*{2cm}\ue^{\frac{\ui}{\hbar}(\phi^{II}_{\boldsymbol{\alpha}}(t,\hat{x}_m,\hat{y}_n,\xi,\eta)-
   \phi^I_{\boldsymbol{\alpha}'}(t,\hat{x}_m,\hat{y}_n,\eta'))}\ud\xi\ud\eta
    \ud\eta'\ud t\\
 &=\frac{1}{2\pi\hbar}\int_{\Lambda_{\boldsymbol{\alpha}'\boldsymbol{\alpha}}}
    g_{\boldsymbol{\alpha}',k',\hbar}
    \ue^{\frac{\ui\pi}{4}f'_{\boldsymbol{\alpha}}}|\partial_\xi x|^{-\frac{1}{2}}
    e_{\boldsymbol{\alpha},k,\hbar}\ud t\ud\xi'\ud\eta' + O\bigl(\hbar^{\infty}\bigr),
\end{split}
\end{equation}
where $g_{\boldsymbol{\alpha}',k',\hbar}$ and $e_{\boldsymbol{\alpha},k,\hbar}$ are 
coefficients of half-densities on  $\Lambda_{\boldsymbol{\alpha}'\boldsymbol{\alpha}}$ 
possessing asymptotic expansions in $\hbar$ that are uniform in $t\in[0,T]$. 
In particular,
\begin{equation}
\label{q96b}
\begin{split}
e_{\boldsymbol{\alpha},k,\hbar}(\lambda)
  &\sim \left(\ui\hbar\mathcal{L}_{X_{H_{\rm ext}}}
     {a}_{\boldsymbol{\alpha},k}(\lambda) +
    \sum_{n\geq 2}\hbar^n e_{\boldsymbol{\alpha},k,n}(\lambda)\right),
\end{split}
\end{equation}
with $\lambda=(t,\xi',\eta')\in\Lambda_{\boldsymbol{\alpha}'\boldsymbol{\alpha}}$. 
Moreover,
\begin{equation}
\label{q132}
f'_{\boldsymbol{\alpha}} := \sgn\Hess_{y,\xi,\eta}\phi^{II}_{\boldsymbol{\alpha}}=
\sgn\begin{pmatrix}
0 & 0 & -1\\
0 & -\partial^2_{{\xi}}S^{II}_{\boldsymbol{\alpha}} 
& -\partial_{{\xi}{\eta}}S^{II}_{\boldsymbol{\alpha}}\\
-1 & -\partial_{{\xi}{\eta}}S^{II}_{\boldsymbol{\alpha}}
& -\partial^2_{{\eta}}S^{II}_{\boldsymbol{\alpha}}
\end{pmatrix}.
\end{equation} 
\end{lemma}
\begin{proof}
We again employ the same strategy as in the proofs of Lemmata~\ref{q93}--\ref{q121}.
The sums over $n$ and $m$ are evaluated with the help of Lemma~\ref{a1} and, 
with some modifications, this gives a result similar to \eqref{q200}. Here we 
apply the stationary phase theorem to the integrations over the variables 
$(y,\xi,\eta)$. The conditions determining stationary points of the phase 
function are 
\begin{equation}
\label{q136}
\begin{aligned}
\partial_y(\phi^{II}_{\boldsymbol{\alpha}}-\phi^{II}_{\boldsymbol{\alpha}'}) = 0\quad
&\Leftrightarrow \quad\eta'=\eta+s\ell_\xi\\
\partial_{{\eta}}(\phi^{II}_{\boldsymbol{\alpha}}-\phi^{II}_{\boldsymbol{\alpha}'}) = 0\quad
&\Leftrightarrow\quad\hat{y}^{\boldsymbol{\alpha}}=
\partial_{\eta}S^{II}_{\boldsymbol{\alpha}}(t,\xi,\eta)\\
\partial_\xi(\phi^{II}_{\boldsymbol{\alpha}}-\phi^{II}_{\boldsymbol{\alpha}'}) = 0 \quad
&\Leftrightarrow\quad\hat{x}^{\boldsymbol{\alpha}}=
\partial_\xi S^{II}_{\boldsymbol{\alpha}}(t,{\xi},{\eta}).
\end{aligned}
\end{equation}
These equations can be solved for $x(t,\xi',\eta')$ and $y(t,\xi',\eta')$. 
They determine the set of stationary points that can either be given in the form 
\begin{equation}
\label{q98a}
\bigl(t,x(t,\xi',\eta')-r\ell_x,y(t,\xi',\eta')-r\ell_x,\xi'-s\ell_\xi,\eta'-s\ell_\xi\bigr)
\in\mathfrak{V}^{II}_{\boldsymbol{\alpha}},
\end{equation}
or
\begin{equation}
\label{q99a}
\bigl(t,x(t,\xi',\eta'),y(t,\xi',\eta'),\eta')\in
\mathfrak{V}^I_{\boldsymbol{\alpha}'}.
\end{equation}
In order to apply the stationary phase theorem we still need the
determinant of the Hessian as given in \eqref{q132}, i.e.
\begin{equation}
\label{q137}
\det\Hess_{y,\xi,\eta}(\phi^{II}_{\boldsymbol{\alpha}}-\phi^{II}_{\boldsymbol{\alpha}'}) = 
\partial^2_\xi S^{II}_{\boldsymbol{\alpha}}=\partial_{\xi}x.
\end{equation}
We recall that we have chosen the partitions such that $\partial_{\xi}x\neq 0$.
\end{proof}
\subsection{Semiclassical construction}
\label{q159}
We now have all the means necessary for the construction of a discrete scFIO as 
in Definition~\ref{q58} that satisfies the estimate \eqref{q53} to any desired 
(positive) power in $\hbar$. In order to achieve this, Lemma~\ref{q50} implies 
that we need an estimate
\begin{equation}
\label{e1}
\int_0^T\tr R(t)^\ast R(t)\ud t = O(\hbar^\beta)
\end{equation}
with $\beta>3$. To get this estimate, we perform the double sum over
$\boldsymbol{\alpha}$ and $\boldsymbol{\alpha'}$ in \eqref{q156} by first fixing 
$\boldsymbol{\alpha'}$ and summing over $\boldsymbol{\alpha}$, and subsequently
summing over $\boldsymbol{\alpha'}$. 

In the following approach the Lemmata \ref{q93}--\ref{q130} will be used
for the summation over $\boldsymbol{\alpha}$, and this should only involve the 
densities $e_{\boldsymbol{\alpha},k,\hbar}$. However, the Lemmata \ref{q121} and 
\ref{q130} show that for fixed ${\boldsymbol{\alpha'}}$ this is not necessarily 
possible as the phases have to be chosen correctly. In order to make such an 
approach feasible we therefore require that $a^{I}_{\boldsymbol{\alpha},k}$ is of the 
form
\begin{equation}
\label{e2}
a^{I}_{\boldsymbol{\alpha},k}(t,x,y,\eta)  =
\kappa_{\epsilon}(y-y_{t,x,\eta})A^I_{\boldsymbol{\alpha},k}(t,x,\eta),
\end{equation}
with $A^I_{\boldsymbol{\alpha},k}$ defined locally on $\Lambda$. Moreover, we require 
that the phases of the amplitude functions $a^{I/II}_{\boldsymbol{\alpha},k}$ are given 
by $\pi/2$ times the Maslov index \eqref{q17}, such that
\begin{equation}
\label{e3}
a^{I/II}_{\boldsymbol{\alpha},k}\ue^{-\frac{\ui\pi}{2}\mu(\gamma(\lambda))}\in\rz.
\end{equation}
For this to be possible we need to investigate the phases $f_{\boldsymbol{\alpha}}$ 
and $f'_{\boldsymbol{\alpha}}$ in \eqref{q122} and \eqref{q131} in combination 
with an additional phase coming from the Maslov bundle contribution. 
\begin{lemma}
\label{q139}
Let $\lambda\in\Lambda_{\boldsymbol{\alpha}'\boldsymbol{\alpha_1}}\cap
\Lambda_{\boldsymbol{\alpha}'\boldsymbol{\alpha_2}}$, and let $f_{\boldsymbol{\alpha_1}}$ 
and $f'_{\boldsymbol{\alpha_2}}$ be as defined in \eqref{q122} and \eqref{q131},
respectively. Then
\begin{equation}
\label{q140}
-\sgn\Hess\phi^{I}_{\boldsymbol{\alpha_1}}(\lambda)+f_{\boldsymbol{\alpha_1}}(\lambda)=
-\sgn\Hess\phi^{II}_{\boldsymbol{\alpha_2}}(\lambda)+f'_{\boldsymbol{\alpha_2}}(\lambda).
\end{equation}
\end{lemma}
\begin{proof}
We have to prove that
\begin{equation}
\label{q141}
-\sgn\Hess\phi^{I}_{\boldsymbol{\alpha_1}}+\sgn\Hess_{y,\eta}\phi^{I}_{\boldsymbol{\alpha_1}}=
-\sgn\Hess\phi^{II}_{\boldsymbol{\alpha_2}}+\sgn\Hess_{y,\xi,\eta}\phi^{II}_{\boldsymbol{\alpha_2}}.
\end{equation}
$\Hess\phi_{\boldsymbol{\alpha}_2}$ is given by \eqref{q119}, and using \eqref{q123} 
a simple calculation shows that
\begin{equation}
\label{q142}
\sgn\Hess\phi^{II}_{\boldsymbol{\alpha_2}}=0=\sgn\Hess_{y,\eta}\phi^{I}_{\boldsymbol{\alpha_1}}.
\end{equation} 
We now compare \eqref{q101} in the proof of Lemma \ref{q93} with \eqref{q132} 
in the proof of Lemma \ref{q130}. Since 
$\lambda\in\Lambda_{\boldsymbol{\alpha}'\boldsymbol{\alpha}}$, in a neighbourhood of 
$\lambda$ there is a bijection between $\partial_\xi$ and $\partial_x$ given by
\begin{equation}
\label{q143}
\partial_\xi=\left(\partial_\xi x\right)\partial_x.
\end{equation}
In combination with Lemma \ref{q168} this shows that there exists a hermitian 
matrix $P$ that is a product of a permutation matrix and a diagonal matrix with 
diagonal entries $1$ and $\eqref{q143}$, such that
\begin{equation}
\label{q144}
-\Hess_{y,{\xi},{\eta}}\phi^{II}_{\boldsymbol{\alpha_2}}(\lambda)=
\left(P^\ast\Hess\phi^{I}_{\boldsymbol{\alpha_1}}P\right)(\lambda),
\end{equation} 
where $\Hess\phi^{I}_{\boldsymbol{\alpha_1}}$ is given in \eqref{q101}. Thus, 
\cite[Theorem 3,p.\ 187]{Lancaster:1969} shows that
\begin{equation}
\label{q145}
-\sgn\Hess\phi^{I}_{\boldsymbol{\alpha_1}}(\lambda)=
\sgn\Hess_{y,{\xi},{\eta}}\phi^{II}_{\boldsymbol{\alpha_2}}(\lambda),
\quad\lambda\in\Lambda_{\boldsymbol{\alpha}'\boldsymbol{\alpha}}, 
\end{equation}
which proves the claim.
\end{proof}
We note that an application of Lemma \ref{q139} to the right-hand sides of
\eqref{q122} and \eqref{q131} establishes the desired property that the 
phases appearing in these expressions only depend on $\boldsymbol{\alpha'}$, 
but not on $\boldsymbol{\alpha}$.

In order to facilitate the construction of the semiclassical time evolution 
we introduce the globally defined half-density $|A_{0}|\ud\lambda^{\frac{1}{2}}$ 
as the solution of the the initial value problem 
\begin{equation}
\label{q85}
\begin{cases}
\left(\mathcal{L}_{X_{\rm ext}}|A_0|\right)(\lambda)\ud \lambda^{\frac{1}{2}}=0, 
  & \lambda\in\Lambda,\\
|A_0|(\lambda)\ud\lambda^{\frac{1}{2}}=|\ud t\wedge \ud y\wedge \ud\eta|^\frac{1}{2}, 
  & \lambda\in\Lambda_0.
\end{cases}
\end{equation}
Here $|\ud t\wedge \ud y\wedge\ud \eta|$ is the canonical density 
of the conormal bundle $\Lambda_0$, see \eqref{a2}. The unique solution of 
\eqref{q85} is   
\begin{equation}
\label{q86}
|A_0|\ud\lambda^{\frac{1}{2}}=
\pi^{\ast}\left(|\ud t\wedge \ud y\wedge\ud \eta|^\frac{1}{2}\right),
\end{equation}
where $\pi^{\ast}$ is the pull-back of the projection
\begin{equation}
\label{q87}
\pi(t,\tau,x,\xi,y,-\eta) = (t,y,-\eta), 
\quad (t,\tau,x,\xi,y,-\eta)\in\Lambda,
\end{equation}
see \cite[Lemma 6.1]{Duistermaat:1975}. The quantity $A_0\ud\lambda^\frac{1}{2}$ 
then involves the Maslov-factor \eqref{e3} in addition to 
$|A_0|\ud\lambda^{\frac{1}{2}}$.

We also define a shifted nearest neighbourhood $N(\lambda)$ of 
$\lambda\in\Lambda$ using the partition
$\left\{\Lambda_{\boldsymbol{\alpha}}\right\}_{\boldsymbol{\alpha}}$,
\begin{equation}
\label{t5}
\lambda'\in N(\lambda)\quad\Leftrightarrow \quad 
\lambda'\in\Lambda_{\boldsymbol{\alpha}'\boldsymbol{\alpha}} 
\quad \mbox{and} \quad 
\lambda\in\Lambda_{\boldsymbol{\alpha}\boldsymbol{\alpha}'}.
\end{equation}
Then
\begin{equation}
\label{t6}
\lambda'\sim\lambda \quad \Leftrightarrow \quad 
\lambda',\lambda\in N(\lambda),
\end{equation}
defines an equivalence relation on 
\begin{equation}
\label{t7aa}\textstyle
\tilde{\Lambda}:=\bigcup_{\boldsymbol{\alpha}}\Lambda_{\boldsymbol{\alpha}}
\setminus\bigcup_{\boldsymbol{\alpha}}\partial\Lambda_{\boldsymbol{\alpha}},
\end{equation} 
and thus yields a disjoint decomposition 
\begin{equation}
\label{t8}
\left\{\tilde{\Lambda}_{\boldsymbol{\beta}}\right\}_{\boldsymbol{\beta}}
\end{equation}
of $\tilde{\Lambda}$. The boundaries in \eqref{t6} can be neglected 
since they do not contribute to the integral estimates in the following. 
\begin{prop}
\label{q146}
For every $\beta>0$ there exists a discrete scFIO $U_{\rm scl}(t)$ in
the sense of Definition~\ref{q58} that satisfies the assumptions of
Lemma~\ref{q50}.  Away from caustics it has the following Van Vleck
form
\begin{equation}
\label{q146a}
{U_\mathrm{scl}(t)}_{m_Nn_N}=\frac{1}{\sqrt{2\pi\hbar}}\frac{\ell_x}{N}\sum_{\alpha_N}
\left|\det\left(\frac{\partial^2W}{\partial x\partial y}(\lambda_{\alpha_N})\right)
\right|^{\frac{1}{2}}\ue^{\frac{\ui}{\hbar}W(\lambda_{\alpha_N})+\frac{\ui\pi}{2}\tilde{\mu}(\gamma_{\alpha_N})}
+O(\hbar).
\end{equation}
Here
$\lambda_{\alpha_N}=(t,\tau,x_{m_N},\xi_{\alpha_N},y_{n_N},-\eta_{\alpha_N})$
is in a suitable neighbourhood of the non-caustic point, the sum
extends over all trajectories $\gamma_{\alpha_N}$ connecting
$(y_{n_N},\eta_{\alpha_N})$ with $(x_{m_N},\xi_{\alpha_N})$, and
\begin{equation}
\label{e6}
\tilde{\mu}(\gamma_{\alpha_N})={\mu}(\gamma_{\alpha_N})+
\frac{1}{2}\sgn\Hess\phi_{\alpha}(\lambda_{\alpha_N}).
\end{equation}
\end{prop}
In order to emphasise that matrix indices range from $0$ to $N-1$,
with a direct implication on the points $x_n$ and, therefore, on
classical quantities, we added an index $N$. Note that the functions
$\phi_{\alpha}$ in \eqref{e6} are independent of $N$.
\begin{proof}
We have to prove that Lemma \ref{q50} is satisfied up to an arbitrary but fixed 
$T>0$. We fix $T$ and choose $\mathfrak{V}_{\boldsymbol{\alpha}}$ sufficiently small 
so as to fulfill the assumptions made previously in this section.

The first step is to prove the condition \eqref{q52} imposed on the initial value
$U_{\rm scl}(0)$. We recall that the initial Lagrangian manifold $\Lambda_0$,
see \eqref{a2}, is generated by a function $\phi^I(t,x,y,\eta)$, see \eqref{q8}.
In agreement with \eqref{q86} we choose for the amplitude function in \eqref{q167},
\begin{equation}
\label{q147}
|A^I_k|(0,x,\eta)=
\begin{cases}
1, &k=0\\
0, &k\geq1,
\end{cases}
\end{equation}
and localise
\begin{equation}
\label{179}
|A_{\boldsymbol{\alpha},k}|\ud\lambda^{\frac{1}{2}}:=
\Psi_{\boldsymbol{\alpha}} |A^I_k|\ud\lambda^{\frac{1}{2}}, 
\end{equation}
where $\Psi_\alpha$ is given in \eqref{q22}.

In order to evaluate $U_{\rm scl}(0)$ we recall that
\begin{equation}
\label{q148}
\phi^I (0,x,y,\eta)=(x-y)\eta
\end{equation}
holds. With Lemma~\ref{a1} we then obtain
\begin{equation}
\label{q180}
\int_{\mathfrak{V}_{\boldsymbol{\alpha}}^\eta}\Psi_{\boldsymbol{\alpha}}
\ue^{\frac{\ui}{\hbar}\phi_{\boldsymbol{\alpha}}}\ud \eta
=\sum_{j=0}^{N-1}
\Psi_{\boldsymbol{\alpha}}(0,\tau,\hat{x}_,\eta_j,\hat{y}_n,-\eta_j)
\ue^{\frac{\ui}{\hbar}(\hat{x}_m-\hat{y}_n)\eta_j} + O(\hbar^{\infty})
\end{equation}
where $\eta_j =j\ell_\xi/N+r\ell_\xi$ with some $r\in\gz$. The
condition \eqref{q197} hence implies
\begin{equation}
\label{q150}
U_{\rm scl}(0)_{mm} = 1 + O(\hbar^{\infty}).
\end{equation}
Off the diagonal, for $U\left(0\right)_{mn}$ with $m\neq n$, we note that 
\begin{equation}
\label{q149}
\frac{1}{\hbar}(\hat{x}_m-\hat{y}_n)\eta_j = \frac{2\pi(m-n)j}{N}+
2\pi l\quad \mbox{for some} \quad l\in\gz.
\end{equation}
Therefore,
\begin{equation}
\sum_{j=0}^{N-1}\ue^{\frac{\ui}{\hbar}(\hat{x}_m-\hat{y}_n)\eta_j}=0,
\end{equation}
which implies that
\begin{equation}
\label{q153}
U_{\rm scl}(0)_{mn}=O(\hbar^\infty).
\end{equation}
Hence, the initial condition \eqref{q52} is fulfilled.

The main task in this proof is to obtain the estimate \eqref{q51}. Every 
term $\tr R_{\boldsymbol{\alpha}',k'}(t)^\ast R_{\boldsymbol{\alpha},k}(t)$ can be 
written as one of the four cases described by the Lemmata \ref{q93}--\ref{q130}. 
We evaluate the integrals ocurring in these Lemmata recursively in their orders 
$k$ on every set $\tilde{\Lambda}_{\boldsymbol{\beta}}$, and note that in each of 
the four cases the half-densities $e_{\boldsymbol{\alpha},k,\hbar}$ possess complete 
asymptotic expansions in powers of $\hbar$ with similar leading terms, see 
\eqref{q96}, \eqref{q111}, \eqref{q96a} and \eqref{q96b}. Since 
$\partial_x\xi=(\partial_\xi x)^{-1}$, the additional factors in \eqref{q94} 
and \eqref{q131} reflect the choice of coordinates for the respective 
half-densities $e_{\boldsymbol{\alpha},k,\hbar}$.

Suppose a global half-density $|A_{k}|\ud\lambda^{\frac{1}{2}}$ on $\Lambda$ 
is given. It can be localised in the covering space using the partition of unity 
$\{\Psi_{\boldsymbol{\alpha}},\Lambda_{\boldsymbol{\alpha}}\}_{\boldsymbol{\alpha}}$ 
of \eqref{q196}, as 
\begin{equation}
\label{q162}
|A_{\boldsymbol{\alpha},k}|\ud\lambda^{\frac{1}{2}}:=
\Psi_{\boldsymbol{\alpha}}|A_{k}|\ud\lambda^{\frac{1}{2}}.
\end{equation}
Due to \eqref{q197} we then have 
\begin{equation}
\label{q163}
\sum_{\boldsymbol{\alpha}}\left(\mathcal{L}_{X_{H_{\rm ext}}}|A_{\boldsymbol{\alpha},k}|
\ud\lambda^{\frac{1}{2}}_{\boldsymbol{\alpha}}\right)(\lambda_{\boldsymbol{\alpha}})
=\sum_{\boldsymbol{\alpha}}\left(\Psi_{\boldsymbol{\alpha}}\left(\mathcal{L}_{X_{H_{\rm ext}}}
|A_k|\ud\lambda^{\frac{1}{2}}_{\boldsymbol{\alpha}}\right)\right)(\lambda_{\boldsymbol{\alpha}})
\end{equation}
in a neighbourhood of the points $\lambda_{\boldsymbol{\alpha}}$, where the sum 
in \eqref{q163}, respectively \eqref{q197}, extends over all ${\boldsymbol{\alpha}}$ 
with $\lambda_{\boldsymbol{\alpha}}\in\Lambda_{\boldsymbol{\alpha}}$ corresponding to the 
same point $\lambda$ in the extended torus phase space, i.e., where 
$\check{\lambda}_{\boldsymbol{\alpha}}=\lambda$.

Hence, combining the Lemmata \ref{q93}--\ref{q130} 
with Lemma \ref{q139} shows that for fixed ${\boldsymbol{\alpha'}}$ in 
\eqref{q156} we can extract the coefficient depending on ${\boldsymbol{\alpha'}}$ 
and perform the ${\boldsymbol{\alpha'}}$-summation. This shows that if we start 
with the globally defined half-density $A_{0}\ud\lambda^{\frac{1}{2}}$ of \eqref{q86}, 
using \eqref{q163} we see that when $k=0$,
\begin{equation}
\label{q164}
\sum_{\boldsymbol{\alpha}}\mathcal{L}_{X_{H_{\rm ext}}}\left(|A_{\boldsymbol{\alpha},0}|
\ud\lambda^{\frac{1}{2}}\right)=0.
\end{equation}
Then, for fixed ${\boldsymbol{\alpha'}}$, we sum all half-densities of the 
next order in $\hbar$, except the term coming from the Lie derivatives, and 
denote the resulting half-density on $\Lambda$ by $r_1\ud\lambda^{\frac{1}{2}}$. 
The next step then is to solve the inhomogeneous Cauchy problem, 
\begin{equation}
\label{q165}
-\left(\mathcal{L}_{H_{\rm ext}}|A_{k}|\ud\lambda^{\frac{1}{2}}\right)
+r_{k}\ud\lambda^{\frac{1}{2}}  = 0,
\end{equation}
for $k=1$ by the method of characteristics in analogy to \cite[Theorem
1.4.1]{Duistermaat:1974}. This can be done since the extended
Hamiltonian flow is complete and tangential to $\Lambda$. Again,
\eqref{q162} and \eqref{q163}, as well as an application of
Lemma~\ref{q190}, show that in \eqref{q156} the next order in $\hbar$
(with $\boldsymbol{\alpha'}$ fixed) also vanishes. We proceed in this
recursive manner by solving \eqref{q165} with the method of
characteristics. The emerging half-density $r_k$ is composed of the
half-densities of order $k$ in $\hbar$ in the asymptotic expansions of
$e_{\boldsymbol{\alpha},k,\hbar}$.  Again \eqref{q162} and
\eqref{q163} together with the Lemmata \ref{q93}--\ref{q130} show that
every term of a given order in $\hbar$ is recursively canceled.  An
application of Borel's theorem \cite[Theorem 4.15]{Zworski:2012} then
shows that $U_{\rm scl}$ exists and satisfies the claim.

The Van Vleck expression follows from an application of the stationary phase 
theorem to the $\eta$-integration in \eqref{q167}, see, e.g., 
\cite[p.\,288]{Meinrenken:1992}.
\end{proof} 
%


%
\section{Trace formula}
\label{t1a}
In this section we prove Theorem~\ref{t12}. We recall that the
starting point is the relation \eqref{TFidea} where on the left-hand
side the sum extends over all eigenvalues $E_n$ of the quantum
Hamiltonian $\op_N(H)$ (counted with their multiplicities), and
$\rho\in C^\infty(\rz)$ is a test function with compactly supported
Fourier transform
\begin{equation}
\label{t27}
\hat{\rho}(E):=\int_{\rz}\rho(x)\ue^{-\ui Ex}\ud x.
\end{equation}
The trace formula \eqref{t18} is derived from \eqref{TFidea} by using
the semiclassical expression for $\tr U(t)$, developed in
Section~\ref{q166}, on the right-hand side.

We assume $E$ to be a regular value of the classical Hamiltonian
$H\in C^\infty(\tz)$ so that the energy surface $H^{-1}(E)$ \eqref{t10} 
is a submanifold of the phase space $\tz$. In this setting the sets 
$\mathcal{O}_T$ of fixed points of $\Phi^T$ always satisfy the clean 
intersection property \cite[p.\,280]{Guillemin:2013}. In particular, 
$\mathcal{O}_0\cap H^{-1}(E)=H^{-1}(E)$. The set $\mathfrak{P}_E$ of
periodic orbits of the Hamiltonian flow with energy $E$ consists
of primitive periodic orbits, $p^{\#}$, and their $m$-fold
repetitions, $p$, with periods $t_p=mt_{p^{\#}}$.

Since the energy surface is one dimensional a simple calculation shows 
that when every point on $H^{-1}(E)$ is non-stationary one has
\begin{equation}
\label{t31}
\vol\bigl(H^{-1}(E)\bigr)=\sum_{p^{\#}}t_{p^{\#}},
\end{equation}
where the sum is over all primitive periodic orbits in $H^{-1}(E)$.

Now let $p\in\mathfrak{P}_E$ and denote its lift to the covering 
space by $\hat{p}$. Then, with some obvious abuse of the notation in \eqref{q14}, 
we set for the action of $p$,
\begin{equation}
\label{t11}
W_p:=W+Et_p+s\ell_\xi x,
\end{equation}
with 
\begin{equation}
\label{choice}
t=t_p, \ (y,\eta)=\hat{p}(0) \ \mbox{and} \ (x,\xi)=\hat{p}\left(t_p\right)
\end{equation}
in 
\eqref{q14}. The integer $s$ is determined by $\xi-\eta=s\ell_\xi$. It turns out that 
$W_p$ is independent of the choice $(t_p,-E,x,\xi,y,-\eta)$ satisfying \eqref{choice}.

The trace formula also requires the Conley-Zehnder index $\sigma_p$ of 
periodic orbits $p\in\mathfrak{P}_E$.  For its definition we refer to 
Appendix~\ref{a17}. 

With this information we are ready to prove our main result.
\begin{proof}[Proof of Theorem~\ref{t12}]
Our first observation is that due to Proposition~\ref{q146} we can
approximate the unitary time evolution $U(t)$ semiclassically by
some $U_\mathrm{scl}(t)$ to any desired order $\beta>0$ in the
sense of Lemma~\ref{q50}. Hence
$|\tr(U(t)-U_\mathrm{scl}(t))|=O(\hbar^\infty)$, and we therefore
now study $\tr U_\mathrm{scl}(t)$. In order to obtain an asymptotic
expansion in powers of $\hbar$ we have to calculate the asymptotic
expansions of $\tr U_{\boldsymbol{\alpha},k}$ for every $k\in\nz_0$.

We first consider the case where $\Lambda_{\boldsymbol{\alpha}}$ is
generated by a function $\phi^I_\alpha$, see Section~\ref{q66}. With
$U_{\boldsymbol{\alpha},k}$ of the form \eqref{q167} one finds
\begin{equation}
\label{t14}
\begin{split}
&\frac{1}{2\pi}\int_\rz\tr U_{\boldsymbol{\alpha},k}(t)\hat\rho(t)\
    \ue^{\frac{\ui}{\hbar}Et}\ud t\\
&\quad=\frac{1}{(2\pi)^2\hbar}\frac{\ell_x}{N}\sum_{m=0}^{N-1}\int_\rz\int_\rz 
    a^I_{\boldsymbol{\alpha},k}(t,\hat{x}_m,\hat{x}_m,\eta)\hat{\rho}(t)\ue^{\frac{\ui}{\hbar}
    [\phi^I_{\boldsymbol{\alpha}}(t,\hat{x}_m,\hat{x}_m,\eta)+E t]}\ud\eta\ud t\\
&\quad=\frac{1}{(2\pi)^2\hbar}\int_\rz\int_\rz\int_\rz  
     a^I_{\boldsymbol{\alpha},k}(t,\hat{x},\hat{x},\eta)\hat{\rho}(t)
     \ue^{\frac{\ui}{\hbar}[\phi^I_{\boldsymbol{\alpha}}(t,\hat{x},\hat{x},\eta)
     +s\ell_\xi x+Et]}\ud\eta\ud x\ud t+\Or(\hbar^{\infty}),
\end{split}
\end{equation}
where the last line follows from an application of Lemma~\ref{a1}. It
turns out that the integer $s$ arising from Lemma~\ref{a1} is the same
as in \eqref{t11}. We apply the stationary phase theorem and identify
as the conditions of stationary phase,
\begin{equation}
\label{t16}
\begin{aligned}
&\partial_t\phi^I_{\boldsymbol{\alpha}}=-E\\
&\partial_x\phi^I_{\boldsymbol{\alpha}}+s\ell_\xi=\eta\\
&\partial_\eta\phi^I_{\boldsymbol{\alpha}}+r\ell_x=x.
\end{aligned}
\end{equation}
Here the last identity follows from
$\hat{y}|_{y=x}=\hat{x}+r\ell_x$. These conditions imply that $x$ and
$\eta$ are on a periodic orbit of period $t$ and energy $E$.
When $t=0$ the set of such periodic points is given by the energy
surface $H^{-1}(E)$. When $t\neq 0$, the points are on a
non-trivial periodic orbit $p\in\mathfrak{P}_E$. We define
$p_{\boldsymbol{\alpha}}:=\Lambda_{\boldsymbol{\alpha}}\cap p$ and
choose $\Lambda_{\boldsymbol{\alpha}}$ small enough such that
$p_{\boldsymbol{\alpha}}$ is connected.

Upon applying the stationary phase theorem we are left with an integral over
$p_{\boldsymbol{\alpha}}$, which we parametrise with the variable $\nu$, 
leading to
\begin{equation}
\label{t17}
\frac{1}{2\pi}\int_\rz\tr U_{\boldsymbol{\alpha},k}(t)\,\hat\rho(t)
\ue^{\frac{\ui}{\hbar}Et}\ud t 
\sim\sum_{p_{\boldsymbol{\alpha}}}\frac{\hat{\rho}(t)}{2\pi}
\int_{p_{\boldsymbol{\alpha}}}\psi_{\boldsymbol{\rho}}(\nu)a_{\hbar,k}(\nu)
\ue^{{\frac{\ui}{\hbar}W_p-\ui\frac{\pi}{2}\sigma_p}}\ud\nu.
\end{equation}
Here $a_{\hbar,k}$ is a smooth function on $p_{\boldsymbol{\alpha}}$
with a complete asymptotic expansion in $\hbar$, and
$\psi_{\boldsymbol{\rho}}$ is determined by
$\Psi_{\boldsymbol{\alpha}}$, see \eqref{q22}. The Maslov contribution
$\sigma_p$ was derived in \cite[Theorem 13]{Meinrenken:1994} by means
of the identification outlined in \cite[p.\,137]{Guillemin:2013} and
is independent of $\nu\in p_{\boldsymbol{\alpha}}$; the same is true
for the action $W_p$.

In order to find the leading semiclassical contribution it suffices to
study $\tr U_{\boldsymbol{\alpha},0}$ in \eqref{q167}. The rest
contributes terms of at least $\Or(\hbar)$ to \eqref{t13}. Since there
are only finitely many contributions labelled by $\boldsymbol{\alpha}$
the errors can be uniformly bounded with respect to $\hbar$. For a two
dimensional phase space the condition of regular periodic orbits
amounts to the identity $\det\left|\eins-\ud f_p\right|=1$ in
\cite[Eq.~(11.23)]{Guillemin:2013}, where $f_p$ denotes the Poincar\'e
map along $p$.  Using this, as well as
\cite[pp.\,279\,\&\,282]{Guillemin:2013}, we conclude that the
right-hand side of \eqref{t17} is equal to
\begin{equation}
\label{t19}
\begin{aligned}
&\frac{\hat{\rho}(0)}{2\pi}\int_{H^{-1}(E)}\psi_{\rho}(\nu)\ud\nu\ 
     \bigl(1+O(\hbar)\bigr)\\
&\hspace{0.5cm}+\sum_{p_{\boldsymbol{\alpha}}}\frac{\hat{\rho}(t_p)}{2\pi}
    \int_{p_{\boldsymbol{\alpha}}}\psi_{\boldsymbol{\rho}}(t)
    \ue^{{\frac{\ui}{\hbar}W_p-\ui\frac{\pi}{2}\sigma_p}}\ud t\ \bigl(1+O(\hbar)\bigr),
\end{aligned}
\end{equation}
where the path $p_{\boldsymbol{\alpha}}$ is parametrised by time. 

The same calculations can be done in the case \eqref{q12}. We again have to 
solve \eqref{t16} together with $\partial_\xi\phi^{II}_{\boldsymbol{\alpha}}=x$.
The sum over all $\boldsymbol{\alpha}$ can be performed as $\psi_{\boldsymbol{\rho}}$ 
satisfies the relation \eqref{q197}. This finally proves the theorem.
\end{proof}
%


%
\section{A semiclassical quantisation condition}
\label{sec5}
In this section our intention is to explore to what extent the trace formula
of Theorem~\ref{t12} allows us to characterise individual eigenvalues of the 
quantum Hamiltonian $\op_N(H)$. 

The first step is to observe that Theorem~\ref{t12} can be rewritten as follows.
\begin{lemma}
\label{z10}
With the assumptions of Theorem \ref{t12}, one has, locally uniformly in $r\in\rz$,
\begin{equation}
\label{z11}
\begin{split}
&\sum_n\rho\left(\frac{E_n-E-r\hbar}{\hbar}\right)\\
&\hspace{2cm}=\sum_{k\in\gz}\sum_{p\in\mathfrak{P}^{\#}_E}\rho
     \left(t_{p^{\#}}^{-1}\left(\frac{W_{p^{\#}}}{\hbar}+rt_{p^{\#}}
     -\frac{\pi}{2}\sigma_{p^\#}-2\pi k\right)\right) + O(\hbar).
\end{split}
\end{equation}
\end{lemma}
\begin{proof}
We note that $W_p=kW_{p^{\#}}$, $t_p=kt_{p^\#}$ and $\sigma_p=k\sigma_{p^\#}$
(see \eqref{z20}) if $p$ is a $k$-fold repetition of $p^{\#}$. This allows us to write the 
periodic-orbit sum in \eqref{t18} as a double sum over the (finitely many) primitive
periodic orbits and their repetitions.  

Furthermore, a comparison of the left-hand sides of \eqref{t18} and \eqref{z11}
reveals that the latter requires the `classical' energy to be $E+r\hbar$. We therefore
need to evaluate the right-hand side of \eqref{t18} at this energy. Using
\begin{equation}
\label{ActTayl}
W_p(E+r\hbar) = W_p(E) + r\hbar t_p(E) + O(\hbar^2)
\end{equation}
we obtain 
\begin{equation}
\label{tz18}
\sum_{k\in\gz}\sum_{p^{\#}}\hat{\rho}\left(kt_{p^{\#}}\right)t_{p^{\#}}
\ue^{\ui k\bigl(\frac{W_{p^{\#}}}{\hbar}+rt_{p^{\#}}-\frac{\pi}{2}\sigma_{p^{\#}}\bigr)}
+ O(\hbar).
\end{equation}
The remainder estimate depends on $\rho$ but is locally uniform in $r$. 
Applying Poisson summation to the sum over $k$ proves the claim.  
\end{proof}
The right-hand side of \eqref{z11} suggests that in the vicinity of $E$ 
one finds an eigenvalue of $\op_N(H)$, iff 
\begin{equation}
\label{BScon1}
\frac{W_{p^{\#}}}{\hbar}+rt_p-\frac{\pi}{2}\sigma_{p^\#}\approx 2\pi k
\end{equation}
%
In other words, together with \eqref{ActTayl} this would be some form
of a Bohr-Sommerfeld quantisation condition: $E+r\hbar$ is an
approximate eigenvalue, iff
\begin{equation}
\label{BScon2}
\frac{1}{2\pi\hbar}W_{p^{\#}}(E+r\hbar) = k 
+ \frac{1}{4}\sigma_{p^{\#}}(E+r\hbar) + O(\hbar)
\end{equation}
holds. 
In the following we want to explore to what extent such a
relation can be derived from the trace formula. Our approach uses the
tools developed in \cite{Petkov:1998} to estimate eigenvalue
clustering. 

In order to determine the cases where \eqref{BScon1} is fulfilled we now 
introduce some counting measures. We assume that $E$ is regular value of $H$. 
Then the energy surface $H^{-1}(E)$ consists of a finite and disjoint set 
of periodic orbits. In analogy to \cite[p.\,23]{Petkov:1998} we introduce 
the function 
\begin{equation}
\label{301}
Q(r;E,\hbar):=\frac{1}{2\pi}\sum_{p^{\#}\in\mathfrak{P}_E}
\left[\pi-\frac{W_{p^{\#}}}{\hbar}+\frac{\pi}{2}\sigma_{p^{\#}}-rt_{p^{\#}}\right]_{2\pi},
\end{equation}
where $[z]_{2\pi}=z\bmod 2\pi$ such that $-\pi<[z]_{2\pi}\leq\pi$. We remark that, 
as a function of $r$, each term in \eqref{301} jumps in value by one at the 
points in the set
\begin{equation}
\label{302}
\Omega_{p^{\#}}(\hbar):=\left\{r\in\rz; \ \left[\frac{W_{p^{\#}}}{\hbar}
-\frac{\pi}{2}\sigma_{p^{\#}}+rt_{p^{\#}}\right]_{2\pi}=0\right\}.
\end{equation}  
We now define
\begin{equation}
\label{321a}
\begin{split}
N_{\min}(\hbar,E,r)
  &:=\left|\left\{p^{\#}\in\mathfrak{P}_E^{\#};\ \Omega_{p^{\#}}(\hbar)
    \cap\left[-\frac{r}{2},\frac{r}{2}\right]\neq\emptyset\right\}\right|,\\
N_{\max}(\hbar,E,r)
  &:=\left|\left\{p^{\#}\in\mathfrak{P}_E^{\#}; \ \Omega_{p^{\#}}(\hbar)
    \cap\left[-\frac{3r}{2},\frac{3r}{2}\right]\neq\emptyset\right\}\right|,
\end{split}
\end{equation}
and relate these cardinalities to the function \eqref{301}.
\begin{lemma}
\label{320}
If $r>0$ is sufficiently small, then
\begin{equation}
\label{321}
\begin{split}
N_{\min}(\hbar,E,r)
  &=Q\left(\frac{r}{2};E,\hbar\right)
       -Q\left(-\frac{r}{2};E,\hbar\right)+O(r)\\
N_{\max}(\hbar,E,r)
  &=Q\left(\frac{3r}{2};E,\hbar\right)
       -Q\left(-\frac{3r}{2};E,\hbar\right)+O(r).
\end{split}
\end{equation}
\end{lemma}
\begin{proof}
We recall that $|\mathfrak{P}_E^{\#}|$ is finite since $E$ 
is a regular value. Furthermore, $r$ sufficiently small means that 
$\Omega_{p^{\#}}(\hbar)\cap[-3r/2,3r/2]$ has at most one element. We then
set $N_{\min,p^{\#}}=1$ if $[-r/2,r/2]\cap\Omega(\hbar,p^{\#})\neq\emptyset$ 
and $N_{\max,p^{\#}}=1$ if $[-3r/2,3r/2]\cap\Omega(\hbar,p^{\#})\neq\emptyset$. 
This gives
\begin{equation}
\label{311}
\begin{aligned}
Q\left(\frac{r}{2};E,\hbar\right)-Q\left(-\frac{r}{2};E,\hbar\right)+O(r)
  &=\sum_{p^{\#}\in\mathfrak{P}_E^{\#}}N_{\min,p^{\#}},\\
Q\left(\frac{3r}{2};E,\hbar\right)-Q\left(-\frac{3r}{2};E,\hbar\right)+O(r)
  &=\sum_{p^{\#}\in\mathfrak{P}_E^{\#}}N_{\max,p^{\#}},
\end{aligned}
\end{equation}  
which proves the claim.
\end{proof}
Since the number of primitive periodic orbits is finite we can always
find $\omega_0>0$ and $r_0>0$ such that for all $r\in[-r_0,r_0]$,
$0<\hbar<\hbar_0$ and $p^{\#}\in\mathfrak{P}_E^{\#}$,
\begin{equation}
\label{303}
\left|\frac{W_{p^{\#}}}{\hbar}-\frac{\pi}{2}\sigma_{p^{\#}}+rt_{p^{\#}}\right|>\omega_0. 
\end{equation}
We remark that, in contrast to \cite[Eq.\,(1.9)]{Petkov:1998}, in \eqref{303} 
an absolute value is taken.

We now define a local eigenvalue counting function as
\begin{equation}
\label{304}
N_{E,r}(\hbar):=\left|\{E_n\in\sigma(\op_N(H)); 
\ |E_n-E|<r\hbar\}\right|.
\end{equation}
Here the eigenvalues in \eqref{304} are counted with their multiplicities.  
In view of the Bohr-Sommerfeld condition \eqref{BScon1} 
our aim therefore is to identify situations where, for small $r$, one has
$N_{E,r}(\hbar)\geq 1$. To this end we establish an upper and a 
lower bound for the local eigenvalue counting function. 
\begin{prop}
\label{304a}
There exist $\hbar_0>0$ and $r_0>0$ such that for all 
$0<r<r_0$ and $0<\hbar<\hbar_0$,
\begin{equation}
\label{312}
N_{\min}(\hbar,E,r)\leq N_{E,r}(\hbar)
\leq N_{\max}(\hbar,E,r).
\end{equation}
\end{prop} 
\begin{proof}
In a first step we prove the bounds
\begin{equation}
\label{305}
\begin{split}
&N_{\min}(\hbar,E,r)-C_0r+O_r(\hbar)\\
&\hspace{3cm}\leq N_{E,r}(\hbar)
    -\frac{r}{\pi}\vol\left(H^{-1}(E)\right)\\
&\hspace{6cm}\leq N_{\max}(\hbar,E,r)
+C_0r+O_r(\hbar)
\end{split}
\end{equation}
with some $C_0>0$. In view of Lemma~\ref{320} this can be done very
much in analogy to the proof of \cite[Theorem\ 1.1]{Petkov:1998}.  We
note that in the present case the integrals over the energy surface in
\cite{Petkov:1998} can be carried out since the Liouville measure of a
periodic orbit is known explicitly; it is $\ud t$ when $t$ is the time
coordinate. We also note that in the present case the analogue of
\cite[Theorem\ 5.1]{Petkov:1998} is given by our Theorem~\ref{t12},
and even includes an improved error term.

In order to prove \eqref{312} we note that the first terms in the
first and third line of \eqref{305} are monotonously increasing in
$r$. Since the first terms in each line of \eqref{305} are integers we
can neglect the terms with $C_0$, $\vol(H^{-1}(E))$ and
$O_{r}(\hbar)$ when $\hbar$ and $r$ are sufficiently small.
\end{proof}
Obviously, when under the conditions of Proposition~\ref{304a} we have
that $N_{\min}(\hbar,E,r)=N_{\max}(\hbar,E,r):=N$ then
$N_{E,r}=N$. Moreover, the upper bound for $N_{E,r}$ is easily
obtained by a combination of Proposition~\ref{304a} and Eq.~\eqref{321a}.
\begin{cor}
\label{313a}
With the assumptions of Proposition~\ref{304a} we have
$N_{E,r}\leq\left|\mathfrak{P}^{\#}_{E}\right|$.
\end{cor}
We obtain the following consequence of Proposition~\ref{304a}:
\begin{prop}
\label{313c}
In addition to the assumption of Proposition~\ref{304a}, suppose that
the exact Bohr-Sommerfeld condition 
\begin{equation}
\label{BScon1a}
\frac{W_{p^{\#}}}{\hbar}-\frac{\pi}{2}\sigma_{p^\#}= 2\pi k.
\end{equation}
holds. Then for every $\hbar<\hbar_0$ we have $N_{E,r}\geq 1$.
\end{prop}
\begin{proof}
We have to show that $N_{\min}(\hbar,E,r)\geq 1$ for every
$\hbar<\hbar_0$. But this follows if we insert the Bohr-Sommerfeld
condition \eqref{BScon1a} in \eqref{321a} resp.\ \eqref{302}.
%
\end{proof}
This proposition is the closest one can get towards a Bohr-Sommerfeld
quantisation condition on the basis of the trace formula only. Notice
that our approach never makes use of any eigenfunctions. Therefore, we
have not relied on constructing quasi-modes, which is the usual
approach to Bohr-Sommerfeld conditions, see, e.g., \cite{Charles:2003}
for the case of Toeplitz operators on compact K\"ahler manifolds.


%
\section{Examples}
\label{t1}
In this section we discuss some simple examples. We begin with
classical Hamiltonians that are either of the form $H(x,\xi)=H(x)$,
or of the form $H(x,\xi)=H(\xi)$.

When the classical Hamiltonian is of the form $H(x)$ the equations of
motion are
\begin{equation}
\label{t2}
\dot{\xi}=-\partial_x H =-H', \quad \dot{x}=\partial_\xi H=0.
\end{equation}
The solutions of \eqref{t2} are given by $x(t)=x_0$ and $\xi(t)=\xi_0-tH'$. 
We assume that the energy $E=H(x_0)$ is not critical. A (primitive) 
periodic orbit is given when $\xi(t_{p^\#})=\xi_0 \pm \ell_\xi$. Hence the 
primitive period is $t_{p^\#}=|\ell_\xi/H'(x_0)|$. Moreover, the action is 
$W_p=\ell_\xi x$ and originates from the third term in \eqref{t11}.

The set $\mathfrak{P}_E$ of periodic orbits of energy $E$ is discrete, 
as $E$ is non-critical. Periodic orbits $p\in\mathfrak{P}_E$ can be 
labelled by the points $x_0\in [0,\ell_x)$ such that $H(x_0)=E$.

On the classical side it remains to determine the Conley-Zehnder index 
of a periodic orbit. A short calculation shows that 
$T\Phi_{t}V_M(H^{-1}(E))=V_M(\Phi_{t}(H^{-1}(E)))$ for every $t$
(see Appendix~\ref{a17}). Therefore, we have the extremal case that every 
point in the extended phase space belongs to a caustic 
\cite[p.\,280]{Meinrenken:1992}. From \eqref{a20} we hence conclude that 
the Conley-Zehnder index is equal to zero. A somewhat lengthy but 
straightforward calculation using \cite[Lemma 8.3]{Meinrenken:2000} indeed 
shows that the third and second term in \cite[Proposition 12]{Meinrenken:1994} 
cancel each other in complete agreement with assumption (A2) in
\cite[p.\,9]{Meinrenken:1994}.

With this input one can set up the trace formula \eqref{t18} in the form
\begin{equation}
\label{t5b}
\begin{aligned}
\sum_m\rho \left(\frac{E-E_m}{\hbar}\right)
  &=\sum_{x_0\in H^{-1}([0,\ell_x))}\frac{\ell_\xi}{2\pi |H'(x_0)|}\hat\rho(0) 
         \bigl(1+O(\hbar)\bigr)\\
  &\quad+\sum_{x_0\in H^{-1}([0,\ell_x))}\sum_{k\in\gz\setminus\{0\}}
         \frac{\ell_\xi}{2\pi |H'(x_0)|}\hat{\rho}\left(\frac{k\ell_\xi}{|H'(x_0)|}\right)
         \ue^{\frac{\ui}{\hbar}k\ell_\xi x_0}\bigl(1+O(\hbar)\bigr).
\end{aligned}
\end{equation}
This trace formula can also be proved directly, making use of the explicit
knowledge of the eigenvalues,
\begin{equation}
\label{t4}
E_m=H\lk\frac{\ell_x m}{N}\rk, \quad m=1,\ldots,N,
\end{equation}
and rewriting 
\begin{equation}
\label{t5a}
\sum_{m=1}^N\rho \left(\frac{E-E_m}{\hbar}\right)
=\frac{1}{2\pi}\int_{\rz}\hat{\rho}(t)\sum_{m=1}^{N}
\ue^{\frac{\ui t}{\hbar}\left[E-H\left( \frac{\ell_xm}{N}\right)\right]}\ud t.
\end{equation}
We apply \cite[Theorem 4]{Dixon:1937} to the sum on the right-hand side,
and use the stationary phase theorem for the emerging integrals over a 
variable $x$ in addition to the integral over $t$. The phase functions are
$\phi(t,x)=t(E-H(x))\pm nx\ell_\xi$, so that the condition of stationary 
phase gives $E=H(x)$ in the variable $t$, and $-tH'(x)=\mp n\ell_\xi$ 
in the varibale $x$. The first condition determines the points $x_0$ labelling
(primitive) periodic orbits of energy $E$, and the second condition 
provides the periods of the orbits. Carrying out the stationary phase theorem
eventually gives \eqref{t5b}.

The case of a classical Hamiltonian $H(x,\xi)=H(\xi)$ is very similar.  A 
particular example of this leads to a discretised version of the negative 
Laplacian. As the Weyl symbol of the Laplacian on $\rz$, i.e., the negative
second derivative, is $H_\rz(x,\xi)=\xi^2$, a naive guess of the classical
Hamiltonian leading to a discretised Laplacian would be to restrict this
function to $\fz$. However, on $\tz$ this would correspond to a 
non-continuous function. Continuity could be restored by choosing
$H(x,\xi)=(\xi-\ell_\xi/2)^2$ for $(x,\xi)\in\fz$. Although this function 
is not smooth on $\tz$, an operator $\op_N(H)$ could still be defined using
\eqref{q34} or \cite[Theorem 2.3]{Ligabo:2016}. Its eigenvalues are
\begin{equation}
\label{t4a}
E_m=H\left(\frac{\ell_\xi m}{N}\right)
=\left(\frac{m}{N}-\frac{1}{2}\right)^2 \ell_\xi^2, \quad m=1,\ldots,N,
\end{equation}
and one could use this explicit expression to set up a trace formula.

A reverse approach would be to start with the `natural' discretised Laplacian
as a difference operator, e.g., defined on $\ell_2(\gz)$ as
\begin{equation}
\label{discLap}
(-\Delta f)_n := -(f_{n+1}+f_{n-1}-2f_n),\qquad (f_n)\in\ell_2(\gz).
\end{equation}
In the present context, where $\op_N(H)$ is an operator on $\kz^N$ expressed
in terms of the phase-space translations \eqref{q32}, the closest analogue to 
\eqref{discLap} would be \eqref{discLapintro}, see also \cite[p.\,160]{Faria:2010}. 
It then follows from \eqref{schroedsymb} that in order to realise this operator as 
$\op_N(H)$ one has to choose
\begin{equation}
\label{t5e}
H(x,\xi)=\frac{\ell_\xi^2}{2\pi^2}\left(1-\cos\left(\frac{2\pi \xi}{\ell_\xi}\right)\right),
\end{equation}
which is independent of $x$. Hence, in analogy to \eqref{t4a} the 
eigenvalues of the discretised Laplacian $-\hbar^2\Delta$ are
\begin{equation}
\label{t5d}
E_m:=\frac{\ell_\xi^2}{2\pi^2}\left(1-\cos\left(\frac{2\pi m}{N}\right)\right), 
\quad m=1,\ldots,N.
\end{equation}
Zero is always a non-degenerate eigenvalue. If $N$ is even, the largest 
possible value $\frac{\ell_\xi^2}{\pi^2}$ is also a non-degenerate 
eigenvalue. Every other eigenvalue is two-fold degenerate. 

As a final example we mention (a variant of) the Harper model 
\cite{Harper:1955} with classical Hamiltonian
\begin{equation}
\label{r1}
H(x,\xi)=\cos\left(\frac{2\pi x}{\ell_x}\right)+\cos\left(\frac{2\pi \xi}{\ell_\xi}\right).
\end{equation}
In $\fz$ the critical points of this function are given by $(0,0)$, $(\ell_x/2,\ell_\xi/2)$, 
$(0,\ell_\xi/2)$ and $(\ell_x/2,0)$. At the first point $H$ takes a maximum, at the 
second a minimum and the other two are saddle points. The energy surface 
$H^{-1}(0)$, as seen in $\fz$, is the straight line connecting the two saddle 
points plus its continuation connecting the points $(\ell_x/2,\ell_\xi)$ and 
$(\ell_x,\ell_\xi/2)$. Therefore, every periodic orbit on the torus is a projection of a 
periodic orbit in the covering space. Hence, its action \eqref{t11} has $s=0$
and, in absolute value, is the phase space area enclosed by the orbit. The
Conley-Zehnder index is $\sigma_p=2$ for a (primitive) orbit at positive
energy, and $\sigma_p=-2$ at negative energy.

With all this input one can then evaluate the trace formula \eqref{t18}, as well
as the quantisation conditions discussed in Section~\ref{sec5}. The latter
are rigorous versions of the condition (1.9) in \cite{Harper:1955}, or
(2) in \cite{Gat:2003}.
%


%
\subsection*{Acknowledgements}
SE gratefully acknowledges financial support through a
Postdoktorandenstipendium of the Deutsche Forschungsgemeinschaft,
which allowed him to work at Royal Holloway, University of London,
where major parts of this research were carried out. We thank Omri Gat
for discussions on the Harper model. We thank Steve Zelditch for
pointing out to us related results in the framework of Toeplitz
quantisation, in particular Ref.~\cite{Borthwick:1998}.
\appendix
%
%
\section{Explicit comparison of anti-Wick and Weyl quantisations}
\label{appAW}
Anti-Wick quantisation is an alternative to the Weyl
quantisation \eqref{q34}--\eqref{q32}, and in this appendix we provide
explicit expressions for the matrix elements of anti-Wick
operators; for details see \cite{Bouzouina:1996}.

If $f\in C^{\infty}(\tz)$ is a classical observable, its anti-Wick quantisation is the operator
\be
\label{antiwick}
\op_N^{AW}(f):=\frac{1}{2\pi\hbar}\int_{\fz}f(x,\xi)\,P(x,\xi)\ \ud x\,\ud\xi,
\ee
in $\kz^N$, where 
\be
\label{projector}
P(x,\xi)_{jk}=\sqrt{\frac{1}{\pi\hbar}}\frac{\ell_x}{N} \sum_{n,m\in\mathbb{Z}} 
\ue^{\frac{\ui}{\hbar}\xi\big(\frac{j-k}{N}-(n-m)\big)\ell_x}\ue^{-\frac{1}{2\hbar}
\big[\big(\frac{j\ell_x}{N}-n\ell_x-x\big)^2+\big(\frac{k\ell_x}{N}-m\ell_x-x\big)^2 \big]}
\ee
is a projector on a coherent state localised at $(x,\xi)\in\tz$. (Here we used the
definitions of \cite{Bouzouina:1996} with the choices $\kappa=(0,0)$, $z=\ui$ and 
$(a,b)=(\ell_x,\ell_\xi)$.)
\begin{prop}
\label{matrixel}
Let $f\in C^{\infty}(\tz)$, then the matrix elements of the anti-Wick operator 
$\op_N^{AW}(f)$ are 
\be
\label{melements}
\begin{split}
&{\op_N^{AW}(f)}_{jk}\\
&\hspace{0.25cm}=\sqrt{\frac{1}{\pi\hbar}} \sum_{n\in\mathbb{Z}} \ue^{-\frac{1}{4\hbar} 
    \big(\frac{j-k}{N}-n\big)^2\ell_x^2}\int_{-\infty}^\infty\frac{1}{\ell_{\xi}} \int_0^{\ell_{\xi}} 
    f(x,\xi)\, \ue^{\frac{\ui}{\hbar}\xi\big(\frac{j-k}{N}-n\big)\ell_{x}}\,\ue^{-\frac{1}{\hbar}
    \big[x-\big(\frac{j+k}{2N}\ell_{x}-\frac{n}{2}x\big)\big]^2}\ \ud\xi\,\ud x.
\end{split}
\ee
\end{prop}
\begin{proof}
We use \eqref{projector} in \eqref{antiwick} and interchange integration and summation, 
exploiting the exponential damping term. We then use the identity
\be
\ue^{-\frac{1}{2\hbar}\big[\big(\frac{j\ell_{x}}{N}-n\ell_{x}-x\big)^2+
\big(\frac{k\ell_{x}}{N}-m\ell_{x}-x\big)^2\big]} 
=\ue^{-\frac{1}{\hbar}\big[x-\big(\frac{j+k}{2N}\ell_{x}-\frac{n+m}{2}\ell_{x}\big)\big]^2}
\ue^{-\frac{1}{4\hbar}\big(\frac{j-k}{N}-(n-m)\big)^2\ell_{x}^2}, 
\ee
shift the summation index, $n-m=\mu$, and change variables in the $x$-integration, 
$x\mapsto x+ml_x$. Hence we obtain
\be
\begin{split}
&{\op_N^{AW}(f)}_{jk}=\\
&\frac{1}{\ell_{\xi}}\sqrt{\frac{1}{\pi\hbar}}\sum_{\mu\in\gz}\ue^{-\frac{1}{4\hbar}
    \big(\frac{j-k}{N}-\mu\big)^2\ell_{x}^2}\int_{0}^{\ell_{\xi}}\sum_{m\in\gz}\int_{m\ell_x}^{(m+1)
    \ell_x}\ue^{\frac{\ui}{\hbar}\xi\big(\frac{j-k}{N}-\mu\big)\ell_{x}}\ue^{-\frac{1}{\hbar}
    \big[x-\big(\frac{j+k}{2N}\ell_{x}-\frac{n}{2}\ell_{x}\big)\big]^2}\,f(x,\xi)\,\ud x\,\ud\xi,
\end{split}
\ee
where in the last step we used $f(x+l_x,\xi)=f(x,\xi)$. The summation over $m$ can now be 
carried out, proving the claim.
\end{proof}
More explicit expressions can be obtained when $f$ depends either on only 
$x$ or $\xi$. First assume that $f(x,\xi)=\varphi(x)$, then the $\xi$-integration can be
performed, yielding
\begin{equation}
{\op_N^{AW}(f)}_{jk}=\delta_{jk}\sqrt{\frac{1}{\pi\hbar}} \int_{-\infty}^\infty \varphi(x)\,
\ue^{-\frac{1}{\hbar}\big(x-\frac{j\ell_x}{N}\big)^2}\,\ud x.
\end{equation}
With the Fourier expansion
\begin{equation}
\varphi(x)=\sum_{n\in\gz}\varphi_n\,\ue^{2\pi\ui n\frac{x}{\ell_x}}
\end{equation}
this becomes
\begin{equation}
{\op_N^{AW}(f)}_{jk}=\delta_{jk}\sum_{n\in\gz}\varphi_n\,\ue^{2\pi\ui\frac{jn}{N}}\,
\ue^{-\hbar\frac{\pi^2n^2}{\ell_x^2}}.
\end{equation}
If, however, $f(x,\xi)=\phi(\xi)$, then the $x$-integration can be performed,
\begin{equation}
{\op_N^{AW}(f)}_{jk}=\sum_{n\in\mathbb{Z}}\ue^{-\frac{1}{4\hbar} 
\big(\frac{j-k}{N}-n\big)^2 \ell_x^2}\frac{1}{\ell_{\xi}} \int_0^{\ell_{\xi}}\phi(\xi)\,
\ue^{\frac{\ui}{\hbar}\xi\big(\frac{j-k}{N}-n\big)\ell_{x}}\,\ud \xi.
\end{equation}
With the Fourier series
\begin{equation}
\phi(\xi)=\sum_{m\in\gz}\phi_m\,\ue^{2\pi\ui m\frac{\xi}{\ell_\xi}}
\end{equation}
this simplifies to
\begin{equation}
{\op_N^{AW}(f)}_{jk}=\sum_{n\in\gz}\phi_{k-j+nN}\,\ue^{-\hbar\frac{\pi^2(k-j+nN)^2}{\ell_\xi^2}}.
\end{equation}
As an example, the anti-Wick quantisation of the classical Hamiltonian
\eqref{schroedsymb} with potential $V(x)=\cos(\frac{2\pi x}{\ell_x})$, 
as in the Harper Hamiltonian \eqref{r1}, is
\begin{equation}
{\op_N^{AW}(H)}_{jk}=-\frac{N^2}{\ell_x^2}\left(\delta_{j,k-1}+\delta_{j,k+1}-2\right)\,
\ue^{-\hbar\frac{\pi^2}{\ell_\xi^2}} + \delta_{jk}\,\cos\left(\frac{2\pi j}{N}\right)\,
\ue^{-\hbar\frac{\pi^2}{\ell_x^2}}.
\end{equation}
We contrast this with the Weyl quantisation of the same symbol,
\begin{equation}
{\op_N(H)}_{jk}=-\frac{N^2}{\ell_x^2}\left(\delta_{j,k-1}+\delta_{j,k+1}-2\right) + 
\delta_{jk}\,\cos\left(\frac{2\pi j}{N}\right),
\end{equation}
which is a discretised Schr\"odinger operator in the usual sense.

From this example one concludes that Weyl quantisation allows one to represent
discretised Schr\"odinger operators, or indeed other difference operators, using 
an $\hbar$-independent symbol. If one were to represent the same operator in
the framework of anti-Wick quantisation one would have to use a symbol
with an $\hbar$-expansion in all orders.

\section{A semiclassical summation formula}
\label{20}
In this appendix we prove a technical statement which is applied in
the main part in several places. For that purpose we need the
following preparation where, for simplicitly we denote either of the
ordered pairs $(\ell_x,\ell_\xi)$ or $(\ell_\xi,\ell_x)$ by
$(\ell,\ell^\ast)$.
\begin{lemma}
\label{a1}
Let $a\in\C_0^\infty(0,\ell)$, and let $\phi\in C^{\infty}(\rz)$ be
such that there exists $\nu\in\gz$ with
\begin{equation}
\label{a10}
(\nu-1) \ell^\ast < \phi'(t) < (\nu+1)\ell^\ast.
\end{equation}
Then 
\begin{equation}
\label{a2a}
\frac{\ell}{N}\sum_{n=0}^{N-1}a\bigl(\tfrac{n\ell}{N}\bigl)
\ue^{\frac{\ui}{\hbar}\phi\bigl(\tfrac{n\ell}{N}\bigr)} 
= \int_0^\ell a(t)\ue^{\frac{\ui}{\hbar}(\phi(t)-\nu \ell^\ast t)} \ud t
 + O\bigl(\hbar^{\infty}\bigr).
\end{equation}
\end{lemma}
\begin{proof}
Taking into account that
\begin{equation}
\label{a19a}
\frac{2\pi N}{\ell}=\frac{\ell^\ast}{\hbar},
\end{equation}
and applying the Poisson summation formula, see
e.g.\ \cite[Theorem~4]{Dixon:1937}, gives
\begin{equation}
\label{a3}
\frac{\ell}{N}\sum_{n=0}^{N-1}a\bigl(\tfrac{n\ell}{N}\bigr)
\ue^{\frac{\ui}{\hbar}\phi\bigl(\frac{n\ell}{N}\bigr)} = 
\int_0^\ell a(t)\ue^{\frac{\ui}{\hbar}\phi(t)}\ \ud t 
+ 2\sum_{n=1}^\infty\int_0^\ell a(t)
\ue^{\frac{\ui}{\hbar}\phi(t)}\cos\bigl(\tfrac{n\ell^\ast t}{\hbar}\bigl)\ud t.
\end{equation}
The phase $\phi$ in the first integral has no stationary points in
$\supp a$, whereas the stationary points in the second integral are
determined by $\phi'(t)=\pm n\ell^\ast$. Due to \eqref{a10} this gives
$n=|\nu|$, and therefore only one term in the sum over $n$ exceeds
$O(\hbar^{\infty})$. Hence \eqref{a2a} follows.
\end{proof}
This Lemma is an essential ingredient needed to prove the following
statement.
\begin{prop}
\label{a6}
Let $H\in C^\infty(\rz^2)$ be periodic with respect to
$\ell_x\gz\oplus \ell_\xi\gz$, and let $h\in C^\infty(\mu
\ell_\xi,(\mu+1)\ell_\xi)$, where $\mu\in\gz$. Define
\begin{equation}
\label{a9}
\varphi(x,\xi) = x\xi-h(\xi),
\end{equation}
and assume that $A\in C^\infty(\rz)$ is compactly supported in
$\bigl(\mu \ell_\xi,(\mu+1)\ell_\xi\bigr)$. Moreover, let
$\kappa_{\epsilon}$ be of the form $\eqref{q78}$ and define
\begin{equation}
\label{a11}
a(x,\xi) = \kappa_{\epsilon}\bigl(x-h'(\xi)\bigr)a(\xi).
\end{equation}
One then obtains for every fixed $l\in\gz$
\begin{equation}
\label{a7}
\begin{split}
&\sum_{m,\atop|l-m|\leq N^{\frac{1}{M}}}\frac{\ell_x}{N}\frac{1}{\ell_\xi}\int_0^{\ell_\xi}H(x_l,\eta)
   \ue^{\frac{\ui}{\hbar}(x_l-x_m)\eta}\ud\eta\int_{\mu \ell_\xi}^{(\mu+1)\ell_\xi}a(x_m,\xi)
   \ue^{\frac{\ui}{\hbar}\varphi(x_m,\xi)}\ud\xi \\
&\hspace{2.5cm}=\frac{\ell_x}{N}\int_0^{\ell_\xi}H(x_l,\xi)A(\xi)
   \ue^{\frac{\ui}{\hbar}\varphi(x_l,\xi)}\ud\xi + O\bigl(\hbar^{\infty}\bigr).
\end{split}
\end{equation}
where the error estimate depends uniformly on $l$.
\end{prop}
\begin{proof}
We first notice that the support of $\kappa_{\epsilon}(x-h'(\xi))$ is
compact, and therefore the function is bounded. Thus the second
integral on the left-hand side in \eqref{a7} is uniformly bounded with
respect to $x_m$ and $\hbar$. The first integral is a Fourier
integral, and since $H$ is smooth, it is of size $O(|l-m|^{-\infty})$. 
Hence, we can add further terms to the series such
that it includes exactly $N$ consecutive integers $m_0,\dots,m_0+N-1$
centered around $l$, causing an error of size $O(\hbar^\infty)$.

Since $\eta\in[0,\ell_\xi]$ and $(\mu-1)\ell_\xi<\xi<(\mu+1)\ell_\xi$,
we have $(\mu-2)\ell_\xi<\xi-\eta<\mu \ell_\xi$. Thus we can apply
Lemma~\ref{a1} with $\nu=\mu-1$ and obtain
\begin{equation}
\label{a20a}
\begin{split}
&\sum_{m,\atop|l-m|\leq N^{\frac{1}{M}}}\frac{\ell_x}{N}\int_0^{\ell_\xi}H(x_l,\eta)
   \ue^{\frac{\ui}{\hbar}(x_l-x_m)\eta}\ud\eta\int_{\mu \ell_\xi}^{(\mu+1)\ell_\xi}a(x_m,\xi)
   \ue^{\frac{\ui}{\hbar}\varphi(x_m,\xi)}\ud\xi \\
&\hspace{1.5cm}=\int_{\rz}\int_0^{\ell_\xi}\int_{\mu \ell_\xi}^{(\mu+1)\ell_\xi}a(x,\xi)H(x_l,\eta)
   \ue^{\frac{\ui}{\hbar}\phi_l(x,\xi,\eta)}
   \ud x\ud\xi\ud\eta + O\bigl(\hbar^\infty\bigr),
\end{split}
\end{equation}
where
\begin{equation}
\label{a21v}
\phi_l(x,\xi,\eta)= x\bigl(\xi-\eta-(\mu-1) \ell_\xi\bigr)-h(\xi)+x_l\eta.
\end{equation}
Regarding $\xi$ as a parameter, we can approximate the integral over
$x$ and $\eta$ with the stationary phase theorem. The Hessian of the
phase function satisfies $|\det\Hess\phi_l|=1$ and
$\sgn\Hess\phi_l=0$, and the stationary points are given by $x=x_l$
and $\eta=\xi-(\mu-1)\ell_\xi$. Since by \eqref{a11} the function $a$
is constant with respect to $x$ in an $\epsilon$-neighbourhood of
$\{(h'(\xi),\xi);\ \xi_0\leq\xi\leq\xi_0+\lambda_{\xi}\}$, all
subleading terms of finite order in $\hbar$ vanish. Therefore,
\begin{equation}
\label{a20b}
\begin{split}
 &\int_{\rz}\int_0^{\ell_\xi}\int_{\mu \ell_\xi}^{(\mu+1)\ell_\xi}a(x,\xi)H(x_l,\eta)
   \ue^{\frac{\ui}{\hbar}\phi_l(x,\xi,\eta)}\ud x\ud\xi\ud\eta \\
 &\qquad = 2\pi\hbar \int_0^{\ell_\xi}a(x_l,\xi)H(x_l,\xi-\mu \ell_\xi)
           \ue^{\frac{\ui}{\hbar}\phi_l(x_l,\xi,\xi-\mu \ell_\xi)}\ud\xi +
           O\bigl(\hbar^\infty\bigr)\\
 &\qquad = \frac{\ell_\xi \ell_x}{N}\int_0^{\ell_\xi}a(\xi)H(x_l,\xi)
           \ue^{\frac{\ui}{\hbar}[x_l\xi-h(\xi)]}\ud\xi + O\bigl(\hbar^\infty\bigr).
\end{split}
\end{equation}
\end{proof}
\section{The Conley-Zehnder index}
\label{a17}
We outline the definition of the Conley-Zehnder index using results
and notations of \cite{Duistermaat:1976,Meinrenken:1994}.

Let $(E,\omega)$ be a symplectic vector space. Its Lagrangian
Grassmannian $\Lambda(E)$ is the set of all Lagrangian subspaces of
$E$. Given $L_1,L_2,L_3\in\Lambda(E)$, we define on $L_1\oplus
L_2\oplus L_3$ the quadratic form
$Q(x_1,x_2,x_3)=\omega(x_1,x_2)+\omega(x_2,x_3)+\omega(x_3,x_1)$, as
well as its signature $s(L_1,L_2,L_3):=\sgn Q$

Suppose now that $L_1$ and $L_2$ depend on $t\in[a,b]$, and that $L_3$
is transversal (see \cite[p.\,143]{Lee:2013}) to $L_1$ and $L_2$ for
all $t\in[a,b]$. We then define \cite[(5)]{Meinrenken:1994}
\begin{equation}
\label{a19aa}
[L_1:L_2]_{a}^{b}=\frac{1}{2}[s(L_1(a),L_2(a),L_3)-s(L_1(b),L_2(b),L_3)].
\end{equation}
Let $0=t_0<t_1<\ldots<t_{\nu}<\ldots<t_k=T$ be a partition of $[0,T]$
and assume that for every subinterval $[t_{\nu-1},t_{\nu}]$ there
exists $L_{3,\nu}\in\Lambda(E)$ such that $L_1,L_2$ and $L_{3,\nu}$
satisfy the above requirements. We then set
\begin{equation}
\label{a18a}
[L_1:L_2]_0^T=\sum\limits_{\nu=1}^k[L_1:L_2]_{t_{\nu-1}}^{t_\nu}.
\end{equation}
Now let $(M,\omega)$ be a symplectic manifold, with cotangent bundle
$\pi:T^\ast M\to M$. The vertical bundle $V_M$ is defined by
$V_M(p)=\ker\left.\T\pi\right|_{p}$, $p\in M$, where $\T\pi$ is the
differential of the projection $\pi$, and is a Lagrangian submanifold.
We assume that a Hamiltonian $H\in C^\infty(M)$ is given that
generates a Hamiltonian flow $\Phi_t$. We also assume that
$\lambda\in\rz$ is a regular value of $H$ and that $p$ is a periodic
orbit of $\Phi_t$ in $H^{-1}(\lambda)\subset M$ with period $t_p$.

We now assume that $\dim M=2$. Then a periodic orbit $p$ itself is a
Lagrangian submanifold. Using \cite[(A2) p.\,9]{Meinrenken:1994} we
now define the Conley-Zehnder index of a periodic orbit $p$ by
\cite[p.\,10]{Meinrenken:1994},
\begin{equation}
\label{a20}
\sigma_{p}=\left[\T\Phi_{t_p}^{-1}V_M(q):V_M(\Phi_{t_p}^{-1}(q))\right]_0^{t_p},
\end{equation}
where $q$ is an arbitrary point on $p$; the expression \eqref{a20} is
independent of the choice of $q$. In a two dimensional phase space the
Conley-Zehnder index is additive, i.e.,
\begin{equation}
\label{z20}
\sigma_{p}=k\sigma_{p^{\#}},
\end{equation}
if $p$ is a $k$-fold repetition of a primitive periodic orbit $p^{\#}$, see 
\cite[p.\,10]{Meinrenken:1994}.


%
%
{\small
\providecommand{\bysame}{\leavevmode\hbox to3em{\hrulefill}\thinspace}
\providecommand{\MR}{\relax\ifhmode\unskip\space\fi MR }
\providecommand{\MRhref}[2]{%
  \href{http://www.ams.org/mathscinet-getitem?mr=#1}{#2}
}
\providecommand{\href}[2]{#2}

%
\end{document}